\documentclass[aps,pre,twocolumn,floatfix]{revtex4}

\usepackage{latexsym}
\usepackage{amsthm}
\usepackage{amssymb}
\usepackage{amscd}
\usepackage{bm}
\usepackage{amsmath}
\usepackage{mathtools}
\usepackage{float}

\usepackage{graphicx}
\usepackage{subfigure}
\usepackage{placeins}
\usepackage[version=4]{mhchem}

\providecommand{\abs}[1]{\ensuremath{\left\vert#1\right\vert}}

\providecommand{\set}[1]{\ensuremath{\left\{#1\right\}}}
\providecommand{\paren}[1]{( #1 )}
\providecommand{\brac}[1]{\ensuremath{\left[#1\right]}}
\newtheorem{proposition}{Proposition}

\begin{document}

\title{Statistical Topology of Bond Networks \\ with Applications to Silica}

\author{B. Schweinhart}
\email{schweinhart.2@osu.edu}
\affiliation{Department of Mathematics, Ohio State University, Columbus, OH, USA}
\author{D. Rodney}
\email{david.rodney@univ-lyon1.fr}
\affiliation{Institut Lumière Matière, University of Lyon, Villeurbanne, Lyon, FR}
\author{J. K. Mason}
\email{jkmason@ucdavis.edu}
\affiliation{Department of Materials Science and Engineering, University of California, Davis, CA, USA}

\begin{abstract}

Whereas knowledge of a crystalline material's unit cell is fundamental to understanding the material's properties and behavior, there are not obvious analogues to unit cells for disordered materials despite the frequent existence of considerable medium-range order. This article views a material's structure as a collection of local atomic environments that are sampled from some underlying probability distribution of such environments, with the advantage of offering a unified description of both ordered and disordered materials. Crystalline materials can then be regarded as special cases where the underlying probability distribution is highly concentrated around the traditional unit cell. Four descriptors of local atomic environments suitable for disordered bond networks are proposed and applied to molecular dynamics simulations of silica glasses. Each of them reliably distinguishes the structure of glasses produced at different cooling rates, with the $H_1$ barcode and coordination profile providing the best separation. 
\end{abstract}

\maketitle

\section{Introduction}
\label{sec_introduction}
A bond network is one in which atoms connected by covalent bonds form a network that extends throughout the material. The nature of covalent bonding is such that every atom of a given species generally forms the same number of covalent bonds (referred to below as the \emph{valence}), though Fig.~\ref{fig:silica} shows that this does not significantly restrict the overall network connectivity. The purpose of this article is to characterize this connectivity, and to enable quantitative comparison of the connectivity of different bond networks. Since the connectivity is entirely defined by knowledge of the existence of atoms and the bonds connecting them, the resulting analysis does not depend on variations in bond energies or the precise geometries of local atomic environments. That is, a bond network is considered as a graph (in the sense used in discrete mathematics) where the atoms constitute the vertices (possibly labelled by atomic species) and the covalent bonds the edges.

\begin{figure}
\centering  
\subfigure[]{\label{fig:ordered}
\includegraphics[width=0.45\linewidth]{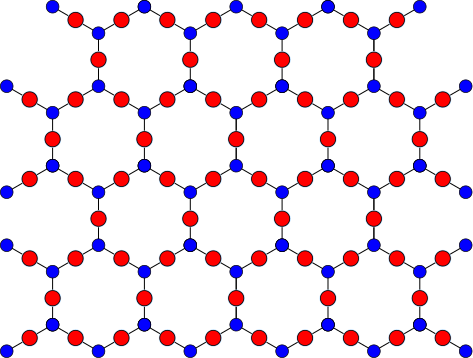}}
\subfigure[]{\label{fig:disordered}
\includegraphics[width=0.49\linewidth]{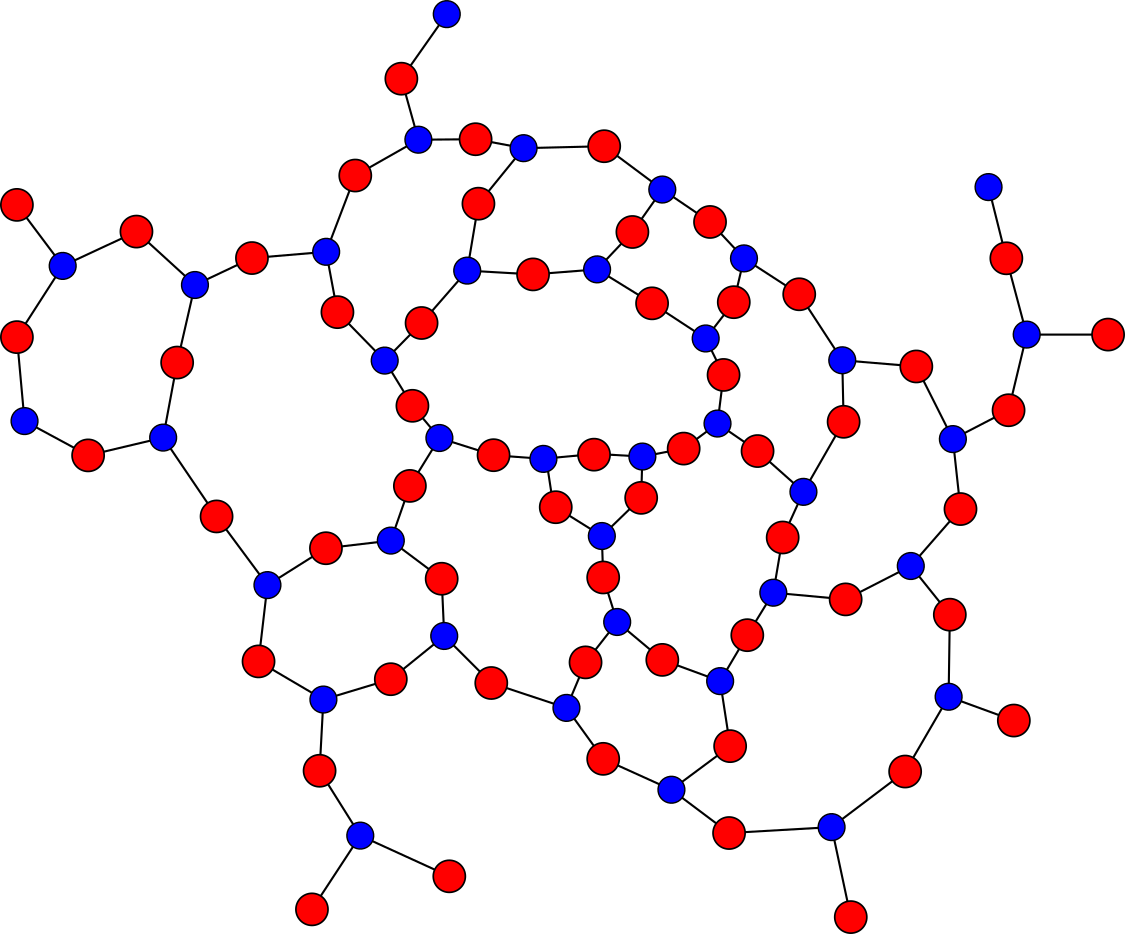}}
\caption{Ordered (a) and disordered (b) bond networks containing 2-valent red atoms and 3-valent blue atoms.
\label{fig:silica}}
\end{figure}

Let the distance between two atoms be defined as the number of edges along the shortest path between the corresponding vertices. A local atomic environment of radius $r$ centered at an atom $v$ is then defined as the subgraph consisting of all atoms within distance $r$ of $v$ and all covalent bonds between atoms in this set, as in Fig.~\ref{fig:environment}. The atom $v$, with a distinguished position at the center of the local atomic environment, is called the \emph{root}. A perfect crystalline solid, by definition, contains a small number of topological types of local atomic environments; considering the polymorphs of silica (\ce{SiO2}) as examples, there is just one environment in cristobalite and two in coesite if the roots are restricted to silicon atoms. By contrast, a disordered solid contains many topological types, though these do not all necessarily have the same probability of occurrence. This suggests that every material with a given chemical composition and processing history has a characteristic probability distribution of local atomic environments, with crystalline solids as special cases where the probability distribution is highly concentrated. Knowledge of this probability distribution would then subsume that of the unit cell, and apply to ordered and disordered bond networks alike. This approach has the further advantage of being sensitive to both continuous and discontinuous variations in the rates of occurrence of local atomic environments, and could therefore be used for, e.g., phase detection, identification of crystal nuclei, or verification that computationally-generated structures correspond to experimental ones.

\begin{figure}
\centering
\includegraphics[width=.45\linewidth]{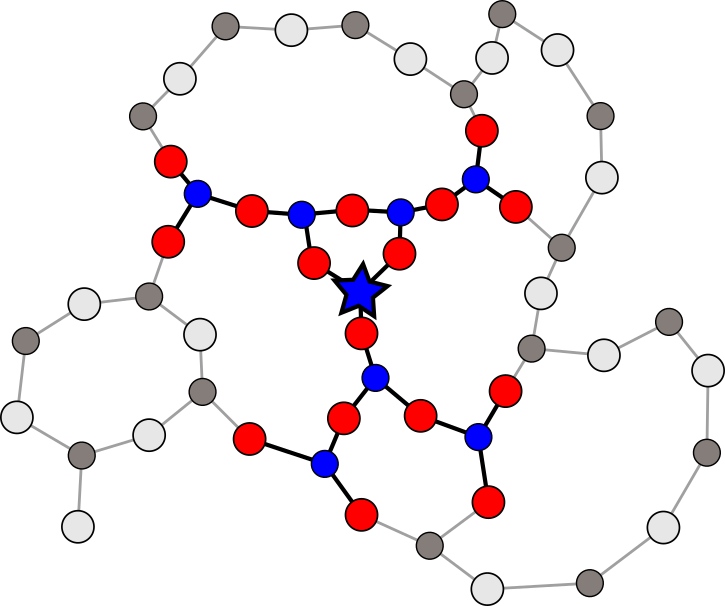}
\caption{A local atomic environment of radius $5$ inside a larger bond network. The root atom is marked by a blue star.
\label{fig:environment}}
\end{figure}

A central concern that has yet to be addressed is precisely which local atomic environments should be considered equivalent and which distinct. We propose that there is no single notion of equivalence that is preferable in all situations. Four possibilities are discussed below in the context of molecular dynamics simulations of silica glasses produced at different cooling rates. For this application, the equivalence relation that allows the atomic networks to be most easily distinguished on the basis of local structural differences is preferred. If instead the relative nucleation rates of competing crystalline phases was the subject of study, then perhaps the one that could most easily distinguish the characteristic probability distributions of different crystalline phases would be more preferable.
 
In previous work~\cite{2012mason,2016schweinhart}, two local environments were considered to be equivalent up to graph isomorphism, effectively employing all of the available topological information to compare local atomic environments. While theoretically satisfying, this can cause the number of equivalence classes (distinct types of local environments) to grow very quickly with the radius of the neighborhood. For example, using this equivalence relation gives $9 \times 10^4$ distinct equivalence classes in a sample of $10^5$ environments in silica glass at a radius sufficient to differentiate the crystalline forms of \ce{SiO2}. This is not a useful summary of the local structure; the true probability distribution of local atomic environments would require an infeasibly large sample population to estimate, and even if known would consider nearly every environment as unique.

This proliferation of equivalence classes is addressed by considering three further equivalence relations which are coarser (encode less information) than graph isomorphism, and effectively assign a summary of topological properties to a local atomic environment. Two are based on classical descriptors, namely, the coordination profile which records the coordination numbers of atoms at each distance from the root, and the primitive ring profile which records the number and length of primitive rings containing the root. The third equivalence relation is the $H_1$ barcode, a novel ring-based descriptor that is defined in terms of the first homology group. The $H_1$ barcode is better at distinguishing local atomic environments in crystalline and glassy silica than the primitive ring profile, particularly when every atom in the network is required to have the standard valence.

There is an extensive literature on local structural descriptors for bond networks (the introduction of Ref.~\cite{2016schweinhart} gives a brief survey). Specifically for covalent glasses, our approach is most similar to the local clusters of Hobbs et al.~\cite{1997hobbs}. The local cluster of an atom is the union of all primitive rings (defined below) containing that atom, and the equivalence relation on local clusters is graph isomorphism. This approach is, however, found to suffer from the same proliferation of equivalence classes described above for the graph isomorphism equivalence relation. Ring statistics are perhaps the most commonly used local structural descriptor for silica networks \cite{1967king}, with a survey of several different classes of distinguished rings given in Ref.~\cite{2010leroux}. Recently, persistent homology was applied to the study silica networks~\cite{2016hiraoka}; while the $H_1$ barcode introduced here can be defined in terms of persistent homology, the approach followed below is quite different.

Previous studies have identified various structural properties of silica glasses that depend on, e.g., the cooling rate from the liquid~\cite{2018deng,2017li,2015koziatek,2013tilocca,1996vollmayr}. Our methodology is applied to simulations of silica glasses quenched from the liquid at three different rates. Overall, glasses that form with faster cooling rates are found to exhibit more disorder, have a higher frequency of rings with $6$ or $8$ atoms ($3$ or $4$ silicon atoms), and have more coordination defects.

\section{Classification of Local Atomic Environments}
\label{sec:class}

Let $V$ be the set of atoms comprising a bond network. A structural descriptor of a bond network is defined as a function $X$ that assigns to the atom $v$ a summary $X\paren{v}$ of the properties of the local atomic environment around $v$. Given a descriptor $X$, atoms $v_1$ and $v_2$ are considered to be equivalent ($v_1 \sim v_2$) if their descriptions are the same ($X\paren{v_1} = X\paren{v_2}$); that is, $X$ induces an equivalence relation on atoms $v \in V$. Given two descriptors $X$ and $Y,$ $X$ is said to be coarser than $Y$ if equivalence under $Y$ implies equivalence under $X$ for all pairs of atoms $v_1$ and $v_2$. Each of the descriptors considered here depends on the radius $r$ of the local environment, and each has the property that $X_r$ is coarser than $X_{r + 1}$ for all $r > 0$ (higher values of the radius provide more information).

Let the \emph{equivalence class} $x$ denote the set of all local atomic environments such that $X(v) = x$ for the root atom $v$. Since the underlying probability distribution on equivalence classes is not usually known for a given descriptor $X$ and atom set $V$, our approach instead uses empirical probability distributions. The empirical probability distribution of $X$ on $V$ assigns to every equivalence class $x$ the probability that $X(v) = x$ for a randomly chosen atom $v \in V$, or
\[P_{X}\paren{x} = (\# v \in V: X(v) = x) / \abs{V}.\]
A descriptor $X$ and two bond networks with atom sets $V_1$ and $V_2$ then result in two discrete probability distributions $P_{X}\paren{x}$ and $Q_{X}\paren{x}$, and a measure of the similarity of these discrete probability distributions could be used to define the similarity of the bond networks (e.g., a metric on discrete probability distributions induces a pseudometric on bond networks). While the symmetrized Kullback--Leibler divergence~\cite{1951kullback} is used for this purpose below, other possibilities include the standard $L_1$ or $L_2$ metrics, or a Wasserstein distance that incorporates information about the geometric similarity of equivalence classes~\cite{2016schweinhart}. 

Three of the four descriptors below encode intuitive topological information about a local atomic environment, though the technical descriptions of the descriptors can be deceptively involved. It is useful to frequently consult examples of local atomic environments and the corresponding descriptors while reading the descriptions below to develop the above-mentioned intuition. One example is provided in Fig.~\ref{fig:example_environment}, and many more of relevance to silica glasses are given in Appendix~\ref{sec:examples}.

\subsection{Graph Isomorphism}

The most detailed of the descriptors uses \emph{graph isomorphism} to construct equivalence classes of local atomic environments. This requires that two local atomic environments $U_1$ and $U_2$ be considered equivalent if there is a function $\phi$ that matches every atom in $U_1$ with an atom in $U_2$ such that (1) $v$ and $\phi(v)$ have the same atomic type for all $v \in U_1$, and (2) $v, w \in U_1$ are bonded if and only if $\phi(v)$ and $\phi(w)$ are as well. The resulting equivalence class specifies all topological information about the local atomic environment. From the standpoint of ring statistics, graph isomorphism controls the number and length of rings as well as their adjacency and distance from the root. The methodology described in~\cite{2016schweinhart} was used to  calculate the corresponding descriptor; that is, a local environment was represented by an adjacency matrix written in canonical form using the software package Nauty~\cite{2014mckay}.

\subsection{The Coordination Profile}

\begin{figure}
\centering
\includegraphics[width=.95\linewidth]{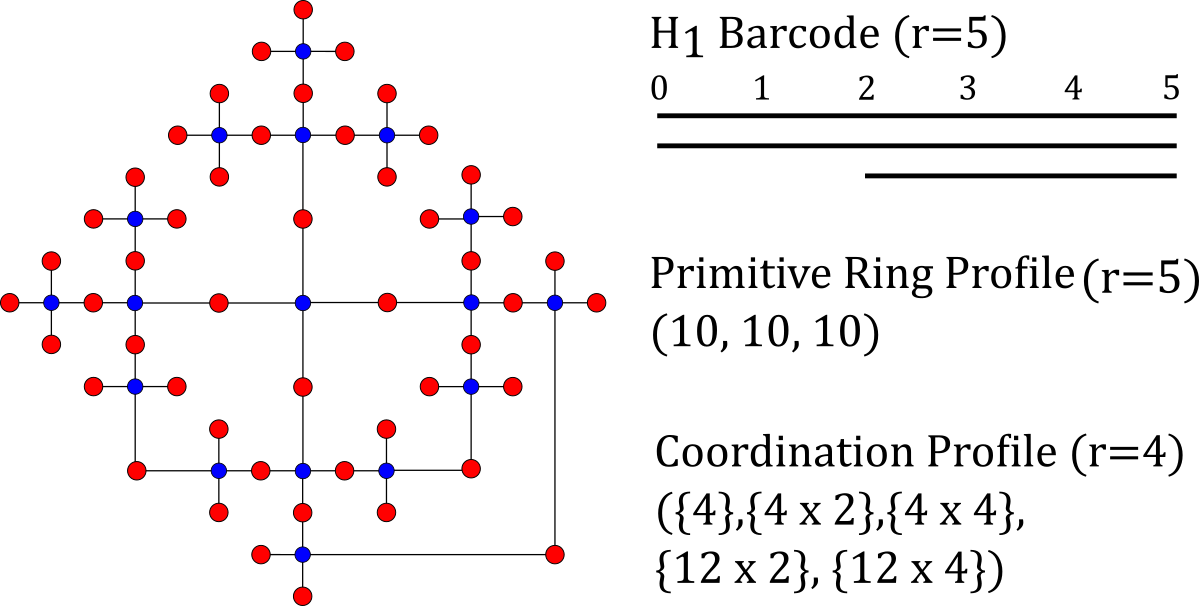}
\caption{An example local atomic environment and its $H_1$ barcode, primitive ring profile, and coordination profile.
\label{fig:example_environment}}
\end{figure}

The \emph{coordination profile} is the simplest descriptor considered here. At radius $r$ this is a vector of $r$ unordered lists where the $i^{\text{th}}$ list gives the coordinations (valences) of all atoms at distance $i^{\text{th}}$ from the root. An example of a local atomic environment and its coordination profile is shown in Fig.~\ref{fig:example_environment}; the root has valence four, each of the atoms in the first neighbor shell has valence two, and the valences alternate with neighbor shell number. Observe that the coordination profile at radius $r$ implicitly includes information about the bonds between atoms in the $r^\text{th}$ and $(r+1)^\text{th}$ shells; this means that the coordination profile at radius $r$ provides a coarser classification than graph isomorphism at radius $r + 1$, but not at radius $r$.

 For a perfectly coordinated silica network---one where every oxygen is adjacent to two silicons and every silicon is adjacent to four oxygens---the coordination profile contains equivalent information to the number of atoms in each neighbor shell around the root. This is known as the \emph{shell count} of the local atomic environment. For example, the radius $5$ shell count of the configuration in Fig.~\ref{fig:example_environment} is $\paren{1,4,4,12,12,33}$.

\subsection{The Primitive Ring Profile}

\begin{figure}
\centering  
\subfigure[]{\label{fig:rings_left}
\includegraphics[width=0.18\linewidth]{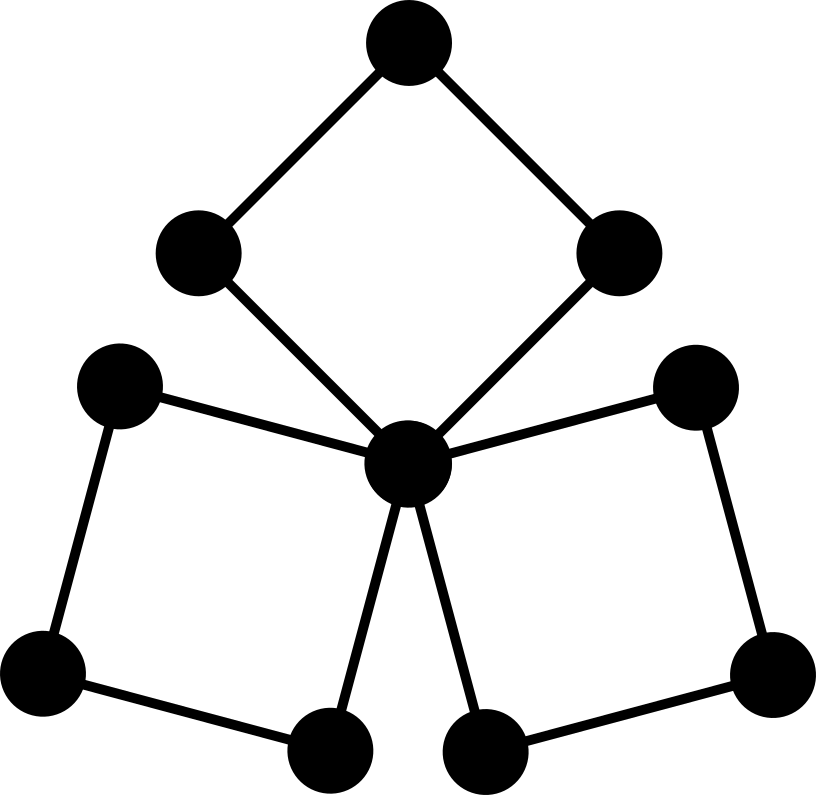}}
\hspace{0.15\linewidth}
\subfigure[]{\label{fig:rings_right}
\includegraphics[width=0.18\linewidth]{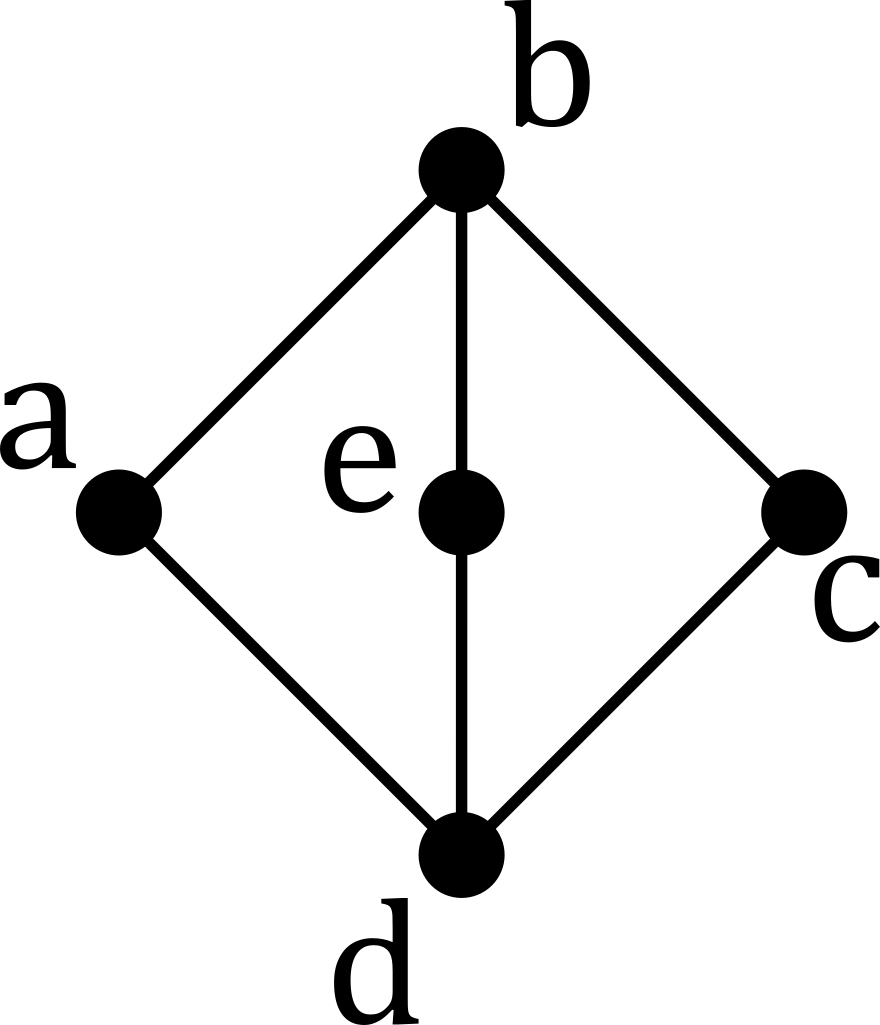}}
\caption{\label{fig:rings} Two configurations with three primitive rings of length four. Whereas (a) has three algebraically independent rings, (b) has only two.}
\end{figure}

A \emph{primitive ring}~\cite{1990marians,1990guttman} is defined as a ring of bonded atoms $\paren{v_1,v_2,\ldots,v_k,v_1}$ such that the shortest path connecting any pair of atoms in the ring is contained within the ring. That is, the shortest path between $v_i$ and $v_j$ for all $1 \leq i < j \leq k$ is either $v_i,v_{i+1},\ldots,v_j$ or $v_j,v_{j+1},\ldots,v_k, v_1,v_2,\ldots,v_{i}$. For example, the configurations in Fig.~\ref{fig:rings} both have three primitive rings. Note that if $v$ is an atom in a bond network, a primitive ring containing $v$ is primitive with respect to the local atomic environment rooted at $v$ if and only if it is primitive in the entire network as well; this is not necessarily true if a ring does not contain $v$. The \emph{primitive ring profile} of radius $r$ of an atom $v$ gives the lengths of all primitive rings which contain $v$ and at most $2 r$ total atoms, as in the example in Fig.~\ref{fig:example_environment}. Unlike the coordination profile, the primitive ring profile stabilizes at a relatively small value of the radius $r$ (there are often alternate paths connecting distant parts of large rings). The algorithm of Yuan and Cormack~\cite{2001yuan} is used to compute primitive rings.

The ring lengths given here are \emph{twice} the usual values in the silica literature where only silicon atoms are reported~\cite{1967king}. For example, a $12$-ring in this article contains $6$ silicon and $6$ oxygen atoms.

\subsection{The $H_1$ Barcode}

The $H_1$ barcode is proposed here as a set of intervals corresponding to algebraically independent rings in a local atomic environment, and is defined in terms of the first homology group. An interval of the form $(j, k)$ with $0 \leq j < k$ corresponds to a ring of unknown length whose atoms are all between distance $j$ and $k$ of the root. An interval $(0, k)$ more specifically corresponds to a ring containing the root atom and either $2 k - 1$ or $2 k$ total atoms. If there are fewer intervals of the form $(0, k)$ than primitive rings passing through the root, then there is an interval of the form $(j, k)$ which encodes information about the relationships between the primitive rings; Fig.~\ref{fig:example_environment} gives an example.

\subsubsection{Definition of Homology}

The definition of homology as it relates to the $H_1$ barcode is more involved than the definitions of the other descriptors considered here; this arises from the requirement that the rings be \emph{algebraically independent}. As motivation for the discussion below, consider that while there is no unique basis for the Euclidean plane, every basis contains precisely two vectors. That is, the space is clearly two-dimensional despite the ambiguity of this choice. The situation is (perhaps surprisingly) similar for rings in a bond network. Consider the three rings in Fig.~\ref{fig:rings_right}; in a sense that is made precise below, any two of these rings can be composed to generate the third in the same way that any two linearly independent vectors in the plane can be regarded as a basis. While the identity of the rings is not well-defined, the number of such algebraically independent rings certainly is well-defined, and the $H_1$ barcode indicates the changes to this number as a function of $r$.

The first homology group is a vector space whose dimension equals the number of algebraically independent rings in a bond network (refer to Ref.~\cite{2002hatcher} for an introduction to general homology theory). If $G$ is a bond network, then the chain group $C_0\paren{G}$ is defined as the vector space of all formal sums $\sum_{i = 1}^k a_i v_i$ where $a_i \in \mathbb{R}$ and $v_i$ is an atom of $G$; this effectively attaches a real number to each atom. Similarly $C_1\paren{G}$ is the vector space of all formal sums $\sum_{i = 1}^k a_i \paren{v_i, w_i}$, where $a_i \in \mathbb{R}$ and $\paren{v_i, w_i}$ is a pair of atoms of $G$ connected by a bond, with the relation $\paren{v_i, w_i} = -\paren{w_i, v_i}$; this effectively attaches a real number to each \emph{oriented} bond. If an element of $C_1(G)$ is regarded as defining the rate of fluid flow along each bond of the network, it is natural to ask the corresponding rate of fluid accumulation around each atom. The linear function $\partial:C_1\paren{G} \rightarrow C_0\paren{G}$ is defined by
\[\partial \bigg[ \sum_{i = 1}^k a_i\paren{v_i, w_i} \bigg] = \sum_{i = 1}^k a_i v_i - a_i w_i\]
and effectively calculates this quantity for every atom of the network simultaneously. The first dimensional homology of $G$ is then the kernel of $\partial$, or the set of all balanced fluid flows:
\[H_1\paren{G} = \set{\sigma \in C_1(G) : \partial(\sigma) = 0}.\]
$H_1\paren{G}$ is generated by oriented rings of $G$, or equivalently every balanced fluid flow can be constructed as a superposition of linearly independent flows around a well-defined number of closed circuits (though the set of closed circuits is not uniquely defined). For example, in Fig.~\ref{fig:rings_right}, $\paren{a,b
}+\paren{b,c}+\paren{c,d}+\paren{d,a}\in H_1\paren{G}$ but $\paren{a,b
}+\paren{b,c}+\paren{c,d}+\paren{a,d}\notin H_1\paren{G}.$ Also, note that if $\sigma_1,\sigma_2,$ and $\sigma_3$ are the three primitive rings
\begin{align*}
\sigma_1=& \paren{a,b}+\paren{b,c}+\paren{c,d}+\paren{d,a} \\
\sigma_2=&\paren{a,b}+\paren{b,e}+\paren{e,d}+\paren{d,a} \\
\sigma_3=&\paren{b,c}+\paren{c,d}+\paren{d,e}+\paren{e,b}
\end{align*}
then there exists the relation
\[\sigma_1=\sigma_2+\sigma_3.\]
That is, as claimed in the opening of this section, only two of the three rings are algebraically independent.

The rank of the vector space $H_1\paren{G}$ gives the number of algebraically independent rings of $G.$ For example, the ranks of the first homology groups of the configurations in Figures~\ref{fig:rings_left} and~\ref{fig:rings_right} are $3$ and $2$, respectively. Practically speaking, the rank of $H_1(G)$ can be efficiently computed using only information about the Euler characteristic and number of connected components of $G$ (two atoms are in the same connected component if there is a path connecting them):
\begin{align*}
\chi\paren{G} &= \#\text{atoms} - \#\text{bonds} \\
&= \#\text{components} - \text{rank}\brac{H_1\paren{G}}.
\end{align*}
Solving for the rank of $H_1(G)$ gives the concise equation
\begin{equation}
\label{eqn_chi}
\text{rank}\brac{H_1\paren{G}} = \#\text{components} - \#\text{atoms} + \#\text{bonds}.
\end{equation}
Finally, we emphasize that the vector space $H_1\paren{G}$ does not have a distinguished basis of ``shortest'' rings; each pair of primitive rings in Fig.~\ref{fig:rings_right} is a viable basis.

\subsubsection{Definition of the $H_1$ Barcode}

\begin{figure}
\includegraphics[width=0.95\linewidth]{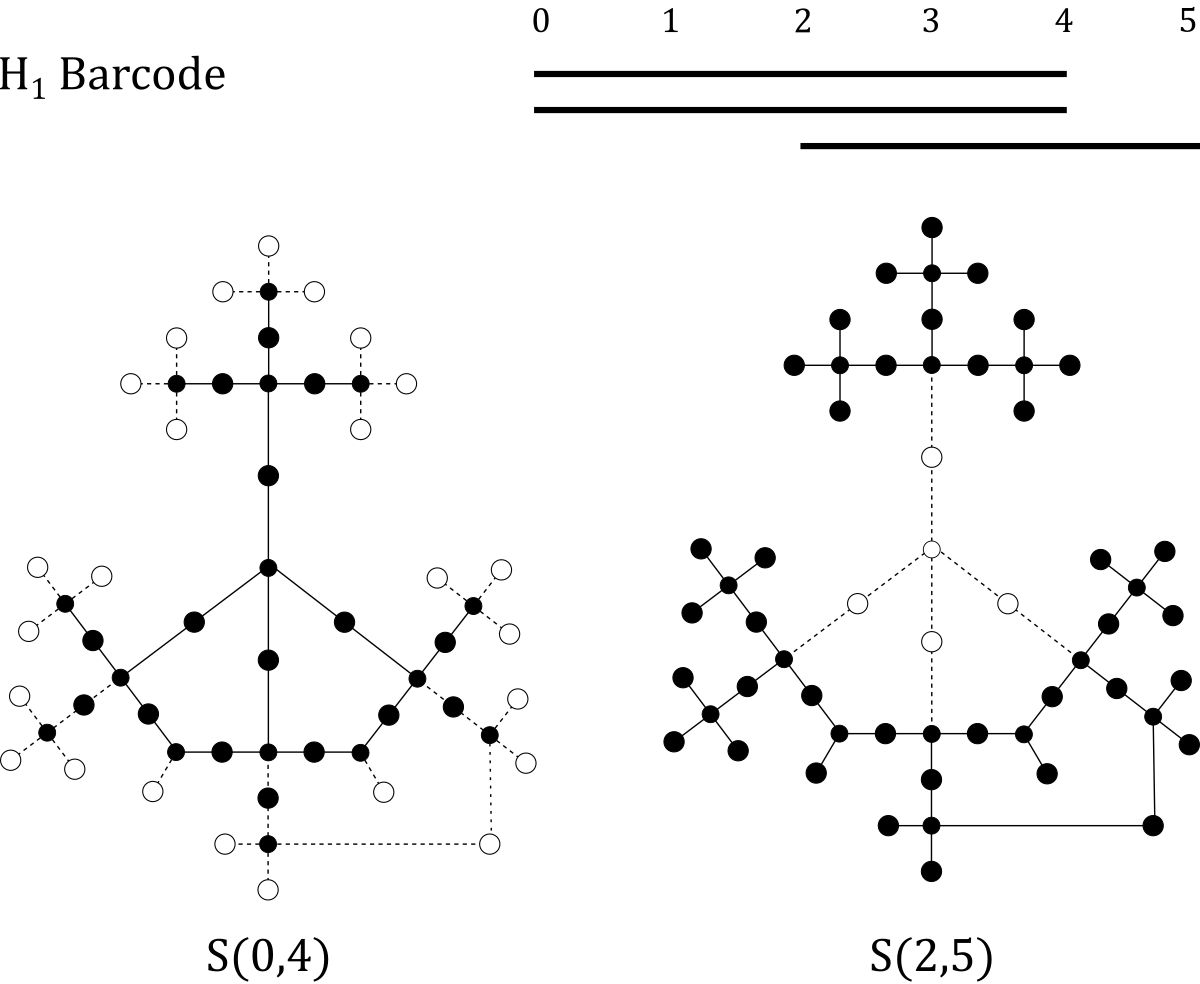}
\caption{\label{fig:shell_example} Two shell annuli of a local atomic configuration, and the corresponding $H_1$ barcode.}
\end{figure}

If $v$ is a root atom in a bond network $G$ and $i\leq j$, the $i,j$-shell annulus $S\paren{i, j}$ at $v$ is the subgraph composed of all atoms between distances $i$ and $j$ from $v$ and the included bonds. $F\paren{i, j}$ is defined as the number of algebraically independent rings of $S\paren{i, j}$, or
\begin{equation}
\label{eq_F}
F\paren{i, j} = \text{rank}\!\set{H_1\!\brac{S\paren{i, j}}}.
\end{equation}
For example, the local atomic environment in Fig.~\ref{fig:shell_example} has $F\paren{0, 5} = 3$, $F\paren{0, 4}=2$ and $F\paren{2, 5} = 1$. If $\paren{i, j} \subseteq \paren{k, l}$, then rings that are algebraically independent in $S\paren{i,j}$ remain algebraically independent in $S\paren{k,l}$, implying that $F\paren{i, j} \leq F\paren{k, l}$. It follows that there is a unique set of intervals $BC$ (the $H_1$ barcode) such that 
\[F\paren{i, j} = \#\set{I \in BC:I \subseteq \paren{i, j}}.\]
The information in $BC$ for a root $v$ is displayed in a barcode, as in Fig.~\ref{fig:shell_example} for which $BC = \set{\paren{0, 4}, \paren{0, 4}, \paren{2, 5}}$.

The $H_1$ barcode could alternatively be defined in terms of the zigzag persistent homology~\cite{2009carlsson} or extended persistent homology~\cite{2009cohensteiner} of the distance function to the root vertex. Those concepts could also be used to define other descriptors that contain different information about the local atomic environment. Here the $H_1$ barcode was computed using M\"{o}bius inversion, as described in Appendix~\ref{appendix:H1}, and (perhaps surprisingly) was faster than an algorithm based on extended persistent homology.


\subsection{Speed of Computation}

\begin{table}
\centering
\begin{tabular}{l|c}
Method & Time (s) \\
\hline
$H_1$ barcode & $12.6$\\
\hline
Primitive ring & $23.6$\\
\hline
Coordination & $1.8$
\end{tabular}
\caption{\label{table:speed}Seconds required to classify $10^5$ atomic environments at radius $6$. The data is from a molecular dynamics simulation of glasses cooled at $5 \times 10^{11} \, \mathrm{K / s}$, which is described in the next section.}
\end{table}

The local atomic environments of a bond network are classified by iterating through a list of all possible root atoms, and computing the selected descriptor for the local atomic environment around each one. The numbers of occurrences of every observed equivalence class are stored in a hash table, with an overall runtime that is linear in the number of root atoms. While the majority of the runtime is spent on computing the descriptors, the asymptotic runtime of these algorithms is not particularly relevant as the local atomic environments contain a small number of atoms. Computations were performed on a $2.3$ GHz AMD Opteron processor with $64$ GB of RAM, and the code was compiled with g++ version 7.3.1 using the -O2 flag.

Table~\ref{table:speed} shows the time required to classify $10^5$ atomic environments at radius $r = 6$. The coordination profile was the fastest, requiring only $1.8$ seconds to compute as this descriptor involved no additional computations once the local atomic environment was known. The $H_1$ barcode required $12.6$ seconds by comparison. The primitive ring profile was the slowest of the three, requiring $23.5$ seconds in the best case. Our implementation used the Yuan and Cormack~\cite{2001yuan} primitive ring algorithm where the global structure of the bond network is used to accelerate the computation of individual profiles. Without this optimization the runtime for the primitive ring profile was $127.1$ seconds, still substantially faster than the algorithm of Hobbs et al.~\cite{1997hobbs} which required $476.7$ seconds.

\section{Applications to Silica}

Silicon dioxide, otherwise known as silica, exists in a variety of crystalline and glassy forms at ambient conditions. Silicon atoms are generally bonded to four oxygen atoms, and oxygen atoms are generally bonded to two silicon atoms, though this is not universally true in poorly relaxed systems, and the coordination of \ce{Si} atoms can be as high as $6$ when under compression \cite{2010benmore}. The four structural descriptors of local atomic environments described in Sec.~\ref{sec:class} were applied to three crystalline forms of silica and to silica glasses generated by molecular dynamics simulations, quenched from the liquid at $5 \times 10^{11} \, \mathrm{K / s}$, $5 \times 10^{12} \, \mathrm{K / s}$, and $5 \times 10^{13} \, \mathrm{K / s}$. An informative descriptor of local structure would ideally distinguish all of these cases. All local atomic environments are rooted at silicon atoms in the following.

The details of the molecular dynamics simulations are described in Appendix~\ref{sec:MD}. For each quench rate, $100$ configurations of $3,000$ atoms ($1,000$ \ce{Si} atoms and $2,000$ \ce{O} atoms) were generated, resulting in data sets of $10^5$ local atomic environments rooted at silicon atoms. The \ce{Si-O} bonds are defined as pairs of \ce{Si} and \ce{O} atoms closer than $2.2$ \AA, the first minimum of the radial distribution function of the silica glasses. This produces environments that are relatively stable to perturbations of the cutoff; if the cutoff is changed by $\pm 0.2$ \AA, $98.3\%$ of the radius $6$ environments in glass quenched at a rate of  $5 \times 10^{11} \, \mathrm{K / s}$ are unchanged. The corresponding percentages for glasses quenched at rates of $5 \times 10^{12} \, \mathrm{K / s}$ and $5 \times 10^{13} \, \mathrm{K / s}$ are $97.0\%$ and $95.0\%,$ respectively. As such, a different choice of cutoff within that range would not significantly affect the analysis below.

Coordination defects are present at low rates in all the quenched silica glasses. As expected~\cite{2015koziatek}, they occur most often at the highest cooling rate for which $1.47\% $ of silicons are $5$-valent, $0.96\%$ of oxygens are $3$-valent, and $0.22\%$ of oxygens are $1$-valent. Although rare from the standpoint of individual atoms, the percentage of radius $6$ environments that contain at least one coordination defect is $24.2\%$, $41.5\%$, and $64.1\%$ for glasses cooled at $5 \times 10^{11} \, \mathrm{K / s}$, $5 \times 10^{12} \, \mathrm{K / s}$ and $5 \times 10^{13} \, \mathrm{K / s}$, respectively. Some of the descriptors considered here are more sensitive to these coordination defects than others. To illustrate this and to give a different perspective on the relative merits of the descriptors, our methodology is applied both to the full samples of $10^5$ atomic environments at each quench rate and to the sub-samples of perfectly coordinated environments. 

Crystalline forms of silica are considered in Sec.~\ref{sec:crystal} before proceeding to the glassy structures. The Shannon entropies of the empirical probability distributions are computed in Sec.~\ref{sec:numClasses}, indicating the extent to which the various descriptors provide informative descriptions of the local structure as a function of radius. Based on this preliminary analysis and on that of the crystalline structures, further analysis is restricted to the $H_1$ barcode, primitive ring profile, and coordination profile at radius $6$, and graph isomorphism at radius $5$. Section~\ref{sec:mutualInformation} uses the mutual information to compare the information provided by the different descriptors. Finally, we compare the local structure of silica glasses produced by different cooling rates in Sec.~\ref{sec:differentCooling}.

\subsection{Crystal Structures}
\label{sec:crystal}

\begin{table}
\centering
\small
\begin{tabular}{l|c|c|c}
 & $H_1$ Barcode & P. Rings & Shell Count \\
\hline
Q & $3\times \paren{0,6},3\times\paren{2,6}$ & $6 \;12\text{-rings}$ & $\paren{1,4,4,12,12,36,30}$\\
\hline
C &  $4\times \paren{0,6}, 5\times\paren{2,6},$ & $12 \;12\text{-rings}$ & $\paren{1,4,4,12,12,36,24}$\\
       &              $4\times\paren{4,6}$                            &                         & \\
\hline
T &  $3\times \paren{0,6}, 7\times \paren{2,6},$ & $12 \;12\text{-rings}$ & $\paren{1,4,4,12,12,36,25}$\\
       &              $1\times\paren{4,6}$                            &                         & 

\end{tabular}
\caption{\label{table:crystal} The $H_1$ barcode, primitive ring profile, and shell count for three crystalline forms of silica: $\alpha$-quartz (Q), cristobalite (C), and tridymite (T).}
\end{table}

Table~\ref{table:crystal} shows the $H_1$ barcode, primitive ring profile, and coordination profiles of local atomic environments rooted at the silicon atoms of three different crystalline forms of silica, namely, $\alpha$-quartz, cristobalite, and tridymite (the absence of coordination defects makes the shell count equivalent to the coordination profile here). None of the descriptors is able to differentiate between the crystal structures for radii $r \leq 5$. The $H_1$ barcodes and coordination profiles of cristobalite and tridymite begin to differ at radius $r = 6$, but the primitive ring profile is unable to distinguish these structures at any radius. This indicates that any descriptor of local atomic environments in silica should be computed at radius $6$ or above to be substantially informative.

\subsection{Shannon Entropy}
\label{sec:numClasses}

An ideally informative descriptor of local atomic environments would retain enough information to differentiate environments, but not so much as to regard each one as unique. Whether a descriptor is too informative can be evaluated by computing the Shannon entropy~\cite{1948shannon} of the corresponding empirical probability distribution. Recall from Section~\ref{sec:class} that, given a bond network with atoms $V$, a descriptor $X$ allows the definition of an empirical probability distribution $P_X$ on equivalence classes $x$. The Shannon entropy associated with the descriptor $X$ is defined as
\[H\paren{X} = -\sum_{x} P_{X}\paren{x} \log[P_{X}\paren{x}].\]
$H(X)$ is minimized by a descriptor that places all atoms in a single equivalence class, resulting in an entropy of $0$. On the other hand, $H(X)$ is maximized if each atom is in an equivalence class unto itself, resulting in an entropy of
\[H\paren{X}=-\sum_{v\in V}\frac{1}{\abs{V}}\log\paren{\abs{V}^{-1}} = \log\paren{\abs{V}}.\]
Normalizing $H\paren{X}$ by $1 / \log\paren{\abs{V}}$ gives scaled entropies between $0$ and $1$, with the property that $X$ is an uninformative descriptor of local structure if the scaled entropy is close to either bound  (though for different reasons).

Table~\ref{table:shannonEntropy} gives the scaled Shannon entropies for the sample of $10^5$ local atomic environments in silica glasses cooled at a rate of $5 \times 10^{11} \, \mathrm{K / s}$. In addition to graph isomorphism, the $H_1$ barcode, the primitive ring profile, and the coordination profile, the primitive cluster of Hobbs et al.~\cite{1997hobbs} is also calculated (truncated by the radius). Recall from the introduction that the primitive cluster is the union of the primitive rings containing the root vertex, and that equivalence classes of primitive clusters are defined using graph isomorphism.

The scaled entropies for graph isomorphism and primitive clusters are already close to $1.0$ for $r = 6$, indicating that they do not provide an informative classification at this radius because the empirical probability distribution cannot be constructed. This is a consequence of the fact that the number of distinct equivalence classes observed approaches the total number of environments in the sample ($9.8 \times 10^4$ different graph isomorphism classes and $8.6 \times 10^4$ different primitive cluster classes). On the other hand, the primitive ring profile, $H_1$ barcode, and coordination profile have scaled entropies ranging from $0.381$ to $0.473$ at $r = 6$, and provide informative classifications at this radius. The primitive ring profile can even be used for larger radii since the number of possible equivalence classes grows more slowly as a function of radius for this descriptor. Since the Shannon entropy is relatively insensitive to the tail of the probability distribution (by design), analogous data for the sub-sample of perfectly coordinated environments is nearly identical.

\begin{table}
\centering
\begin{tabular}{l|c|c|c|c|c}
r & Graph Iso.  &  P. Cluster &  $H_1$ &    P. Rings &  Coordination  \\
\hline 
$4$  &  $0.136$ & $0.101$ &  $0.099$ &  $0.094$ & $0.142$ \\
\hline         
$5$  &  $0.445$ & $0.358$ &  $0.270$ &  $0.226$ & $0.299$ \\
\hline
$6$ &  $0.983$ & $ 0.861 $ & $0.471$ &  $0.381$ & $0.473$ \\
\hline
$7$  &  $1.000$ & $0.999$ &  $0.704$ &  $0.550$ &  $0.678$ \\
\hline
$8$  &  $1.000$ & $-$ &  $0.902$ &  $0.706$ &  $0.846$ \\
\hline
$9$  &  $1.000$ & $-$ &  $0.986$ &  $0.821$ &  $0.965$ \\
\end{tabular}
\caption{\label{table:shannonEntropy} Scaled Shannon entropies of the empirical probability distributions for various descriptors for $4 \leq r \leq 9$. Data is computed for a sample of $10^5$ local atomic environments in silica glasses cooled at $5 \times 10^{11} \, \mathrm{K / s}$.}
\end{table}


\subsection{Mutual Information}
\label{sec:mutualInformation}

A natural question to ask is whether different descriptors encode similar information about local atomic environments. This is measured by means of the uncertainty coefficient $U(X|Y)$ between the empirical probability distributions induced by descriptors $X$ and $Y$; the uncertainty coefficient effectively indicates the amount of information about $P_{X}$ that can be deduced given prior knowledge of $P_{Y}$ \cite{2007press}. The precise definition involves the empirical joint distribution
\[P_{X, Y}\paren{x, y} = (\# v\in V: X(v) = x, Y(v) = y) / \abs{V}.\]
The mutual information~\cite{1948shannon} of $P_X$ and $P_Y$ can then be defined as
\[I\paren{X; Y} = \sum_{x, y} P_{X, Y}\paren{x, y} \log \bigg[ \frac{P_{X, Y}\paren{x, y}}{P_{X}\paren{x} P_Y\paren{y}} \bigg] .\]
Finally, the uncertainty coefficient between $P_X$ and $P_Y$ is the ratio of the mutual information to the Shannon entropy of $P_X$, or
\[U\paren{X|Y} = I\paren{X; Y} / H\paren{X}.\]
$U\paren{X|Y} = 0$ if and only if $P_X$ and $P_Y$ are independent (i.e., none of the information encoded by $X$ is given by $Y$). Conversely, $U\paren{X|Y} = 1$ if and only if $Y\paren{v_1} = Y\paren{v_2}$ implies that $X\paren{v_1} = X\paren{v_2}$ for all $v_1, v_2 \in V$ (i.e., all of the information encoded by $X$ is given by $Y$). Observe that this definition is inherently asymmetric. \begin{table}
\centering
\begin{tabular}{l|c|c|c|c}
 & $H_1,$ 6 & P. Rings, 6 & Coord., 6 & Graph Iso., 5\\
\hline
$H_1,$ 6 & $1.00$ & $0.52$ &  $0.83 $ & $0.62$ \\
\hline
P. Rings, 6 &  $0.65$ & $1.00$ & 0.62 & $0.63$ \\
\hline
Coordination, 6&  $0.83$ & $0.50$  & $1.00$ & $0.61$\\
\hline
Graph Iso., 5&  $0.66$ & $0.54$ & $0.65$ & $1.00$  
\end{tabular}
\caption{\label{table:uncertainty} The uncertainty coefficients of three classifiers at $r = 6$ and graph isomorphism at $r = 5$ applied to a data set of $10^5$ local atomic environments in silica glasses cooled at $5 \times 10^{11} \, \mathrm{K / s}$. }
\end{table} \begin{table}
\centering
\begin{tabular}{l|c|c|c|c}
 & $H_1,$ 6 & P. Rings, 6 & Coord., 6 & Graph Iso., 5\\
\hline
$H_1,$ 6 & $1.00$ & $0.53$ &  $0.82 $ & $0.58$ \\
\hline
P. Rings, 6 &  $0.63$ & $1.00$ & 0.59 & $0.60$ \\
\hline
Coordination, 6&  $1.00$ & $0.60$  & $1.00$ & $0.60$\\
\hline
Graph Iso., 5&  $0.68$ & $0.59$ & $0.57$ & $1.00$  
\end{tabular}
\caption{\label{table:uncertainty_perfect} The same data as in Table~\ref{table:uncertainty} but for the sub-sample of perfectly coordinated environments.}
\end{table} Table~\ref{table:uncertainty} gives the uncertainty coefficients between the $H_1$ barcode ($B_6$), primitive ring profile ($P_6$), and coordination profile ($V_6$) at radius $r = 6,$ and graph isomorphism at radius $r = 5,$ for the data set of $10^5$ local atomic environments in silica glasses cooled at $5 \times 10^{11} \, \mathrm{K / s}$. The uncertainty coefficients $U\paren{B_6|V_6} = 0.83$ and $U\paren{V_6|B_6} = 0.83$ are quite high, and substantially higher than $U\paren{B_6|P_6} = 0.52$ and $U\paren{V_6|P_6} = 0.50$. This can be explained by observing that the $H_1$ barcode and coordination profile consider all atoms in the local atomic environment of radius $6$ rather than just those in the local cluster. What is perhaps surprising is that $U\paren{P_6|B_6} = 0.65$ is not higher. This indicates that the number and length of primitive rings at the root encodes different information than the number and length of algebraically independent rings in the local atomic environment. That the values in the fourth column are all around $0.6$ indicates that a large amount of information is gained by considering atoms in the $6^\text{th}$ neighbor shell. If the $H_1$ barcode and primitive ring profile were instead constructed for $r = 5$ they would be fully explained by graph isomorphism at that radius, resulting in uncertainty coefficients of $1.00$. 

Table~\ref{table:uncertainty_perfect} shows the uncertainty coefficient computations for the sub-sample of $75,849$ perfectly coordinated local atomic environments in glasses cooled at $5 \times 10^{11} \, \mathrm{K / s}$. The most dramatic change is that $U\paren{V_6|B_6} = 1.00$, indicating that the shell count provides a strictly coarser classification (or strictly less information) than the $H_1$ barcode for perfectly coordinated environments. In fact, the shell count provides the same information as the endpoints of the $H_1$ barcode intervals in this case, as shown in Appendix~\ref{sec:perfectH1}.

\subsection{Different Cooling Rates}
\label{sec:differentCooling}
\begin{figure*}
\subfigure[]{\label{fig:H1_freqs}
\includegraphics[width=0.32\textwidth]{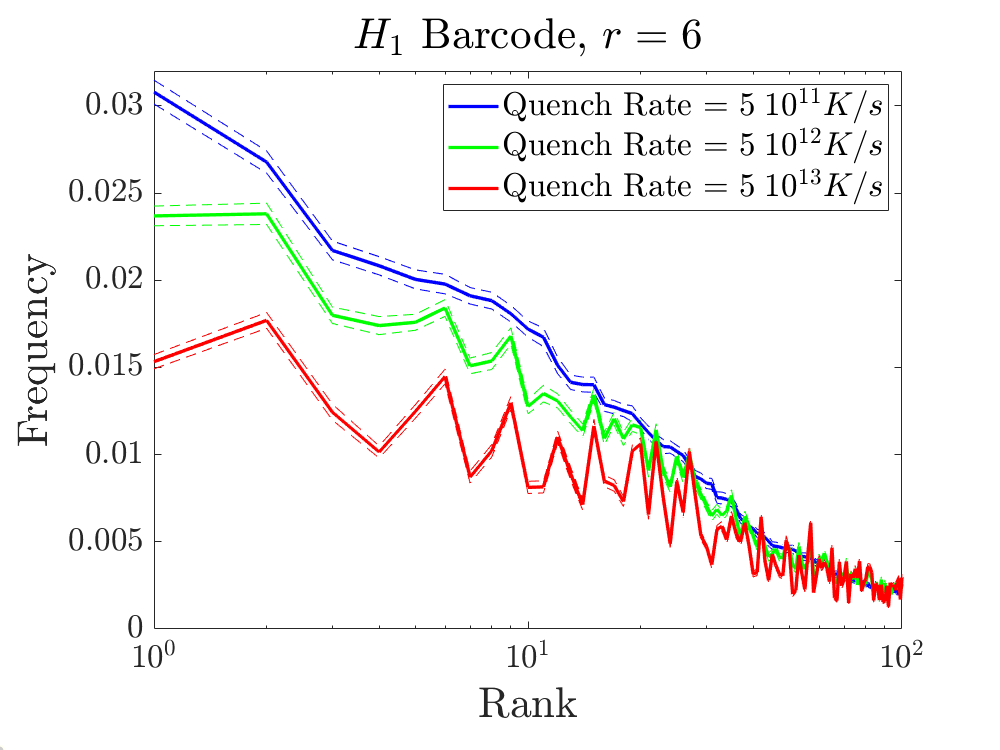}}
\subfigure[]{\label{fig:ring_freqs}
\includegraphics[width=0.32\textwidth]{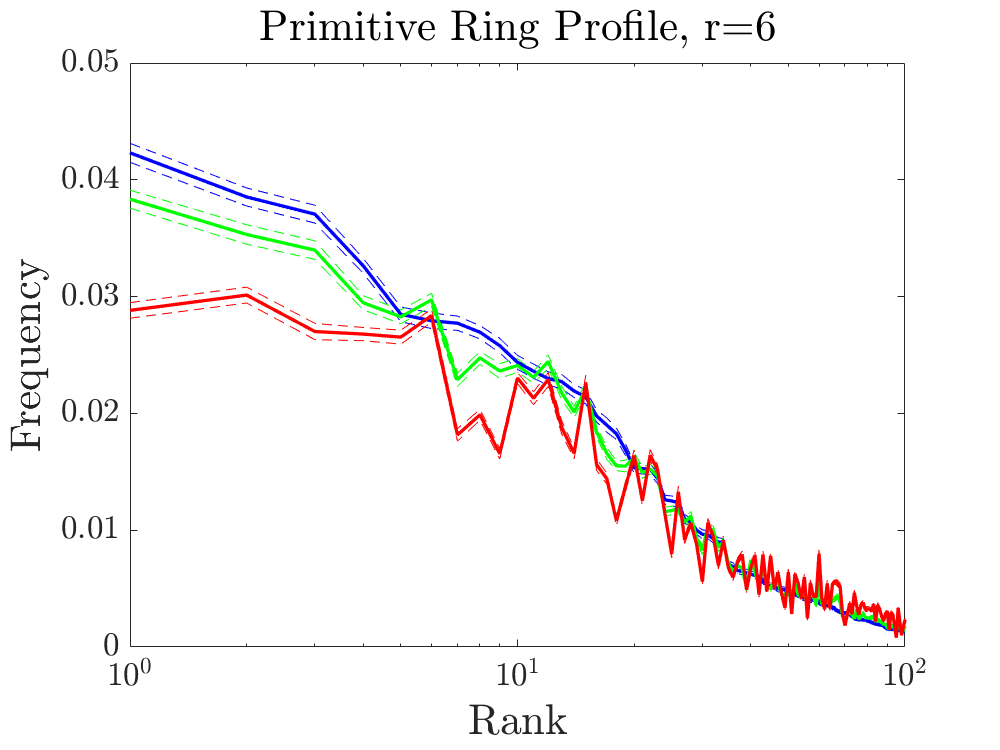}}
\subfigure[]{\label{fig:valence_freqs}
\includegraphics[width=0.32\textwidth]{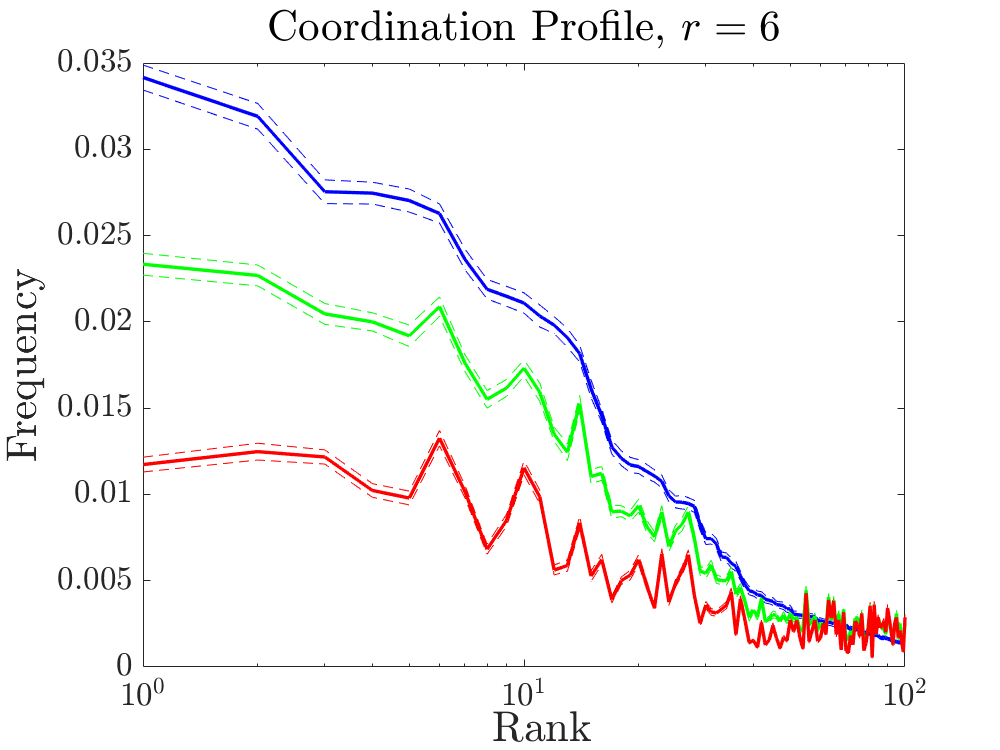}}
\caption{\label{fig:freqs}Empirical probability distributions of equivalence classes in silica glasses produced at different cooling rates. The equivalence classes are ranked by frequency for the glass with the slowest cooling rate, with the frequency of the class on the $y$-axis. The descriptors corresponding to the first ten equivalence classes in each figure are given in Appendix~\ref{sec:examples} (the first column of each of the three tables). The dashed lines indicate the standard error of the frequency estimates.}
\end{figure*}
The initial analysis of silica glasses prepared at different cooling rates involves comparing the shapes of the empirical probability distributions for three different descriptors in Fig.~\ref{fig:freqs}. The rank of an equivalence class in the glass produced with the slowest quench rate is given on the $x$-axis, and the frequency for each of the three cooling rates on the $y$-axis. The blue curve corresponding to the slowest quench rate $5 \times 10^{11} \, \mathrm{K / s}$ is monotonically decreasing, but the green and red curves corresponding to quench rates of $5 \times 10^{12} \, \mathrm{K / s}$ and $5 \times 10^{13} \, \mathrm{K / s}$ are jagged; this is because the ranking of the most common equivalence classes is different in the three cases. While the rank may change, common equivalence classes for one cooling rate tend to be common in the others as well.

In all three figures, the blue curve is highest followed by the green and then the red. That is, the empirical probability distributions for glasses cooled at faster rates are broader, the most common classes are relatively less common, and more probability mass is in the tail of the distribution. In other words, glasses cooled at faster rates exhibit more disorder in their local structure. The corresponding plots for the perfectly coordinated sub-samples (not shown) are similar for the $H_1$ barcode and primitive ring profile, with slightly higher frequencies and slightly less separation between the cooling rates. The coordination profile provides substantially less separation for the sub-sample; this change is quantified below.

\subsubsection{KL Divergence}

The ability of the descriptors to differentiate glasses produced at different quench rates is quantified by the symmetrized Kullback--Leibler (KL) divergence~\cite{1951kullback} between the corresponding empirical probability distributions. The KL divergence $D_{KL}(P_X \| Q_X)$ effectively measures the relative entropy of the probability distribution $P_X$ with respect to $Q_X$, and is perhaps the most natural extension of the Shannon entropy to this context. While the KL divergence is asymmetric, the symmetrized version conforms more to our intuitive notions of a distance (it is a semimetric on probability densities, but a pseudosemimetric on bond networks).

Given a descriptor $X$, two bond networks $V_1$ and $V_2$ give the two empirical probability distributions $P_X$ and $Q_X$ over a single set of equivalence classes $x$. The KL divergence between $P_X$ and $Q_X$ is defined as
\[D_{KL}\paren{P_X \| Q_X} = \sum_{x} P_X\paren{x} \log \bigg[ \frac{P_X\paren{x}}{Q_X\paren{x}} \bigg].\]
This is symmetrized by adding the relative entropy of $P_X$ with respect to $Q_X$ to that of $Q_X$ with respect to $P_X$, or
\[D_{KL}\paren{P_X, Q_X} = D_{KL}\paren{P_X \| Q_X} + D_{KL}\paren{Q_X \| P_X}.\]
Practical calculation of $D_{KL}\paren{P_X, Q_X}$ requires the substitution $q \log\paren{q / p} = 0$ if either $q$ or $p$ is zero, and is justified under two reasonable assumptions: (1) the frequency of an equivalence class in the underlying probability distributions is non-zero in one preparation if and only if it is non-zero in the other, and (2) if an equivalence class is so rare in one preparation so as to not appear in the sample, it is also rare in the other.

Table~\ref{table:KL} shows the symmetrized KL divergences between the empirical probability densities, effectively measuring the separation between the curves in Fig.~\ref{fig:freqs}. Lower values indicate that larger sample sizes are required to reliably differentiate glasses produced at different cooling rates. The coordination profile provides the best differentiation between different cooling rates, followed by the $H_1$ barcode and finally the primitive ring profile. The coordination profile's advantage is, however, attributed to its sensitivity to coordination defects; consider Table~\ref{table:KL_defectFree}, which shows the corresponding data for the sub-samples of perfectly coordinated environments. In this case, the $H_1$ barcode provides the best discrimination between quench rates, followed by the coordination profile and finally the primitive ring profile.

Symmetrized KL divergences were also computed between the empirical probability densities for the full sample of $10^5$ local atomic environments in glasses cooled at $5 \times 10^{13} \, \mathrm{K / s}$ and the sub-sample of $3.6 \times 10^4$ perfectly coordinated environments. The divergence for the coordination profile is very large ($0.657$), and much larger than those given by the $H_1$ barcode ($0.086$) or the primitive ring profile ($0.050$). That is, the coordination profile is quite sensitive to coordination defects, but the other two descriptors are not. Divergences for glasses cooled at slower rates were smaller since perfectly coordinated environments made up a larger proportion of the total sample.

\begin{table}
\centering
\begin{tabular}{l|c|c|c}
Comparison & $H_1$ barcode & P. Rings & Coord. \\
\hline
 $5\times 10^{11},5\times 10^{12}$ & $0.065$ & $0.026$ & $0.162$ \\
\hline
 $5\times 10^{11},5\times 10^{13}$ & $0.288 $ & $0.146$ & $0.655$ \\ 
\hline
 $5\times 10^{12},5\times 10^{13}$ & $0.149$ & $0.066$ & $0.279$ 
\end{tabular}
\caption{\label{table:KL} The symmetrized Kullback--Leibler divergence for three different descriptors with $r = 6$ between empirical probability distributions at different cooling rates.}
\end{table}
\begin{table}
\centering
\begin{tabular}{l|c|c|c}
Comparison & $H_1$ barcode & P. Rings & Coord. \\
\hline
 $5\times 10^{11},5\times 10^{12}$ & $0.044$ & $0.019$ & $0.021$ \\
\hline
 $5\times 10^{11},5\times 10^{13}$ & $0.198 $ & $0.121$ & $0.140$ \\ 
\hline
 $5\times 10^{12},5\times 10^{13}$ & $0.112$ & $0.062$ & $0.069$ 
\end{tabular}
\caption{\label{table:KL_defectFree} The same data as in Table~\ref{table:KL_defectFree}, but for the sub-samples of perfectly coordinated environments (for which the coordination profile gives equivalent information to the shell count).}
\end{table}

\subsubsection{Ranked Examples}
\label{sec:ranked}
\begin{figure*}
\centering  
\subfigure[]{\label{fig:r5_1_a}
\includegraphics[width=0.3\textwidth]{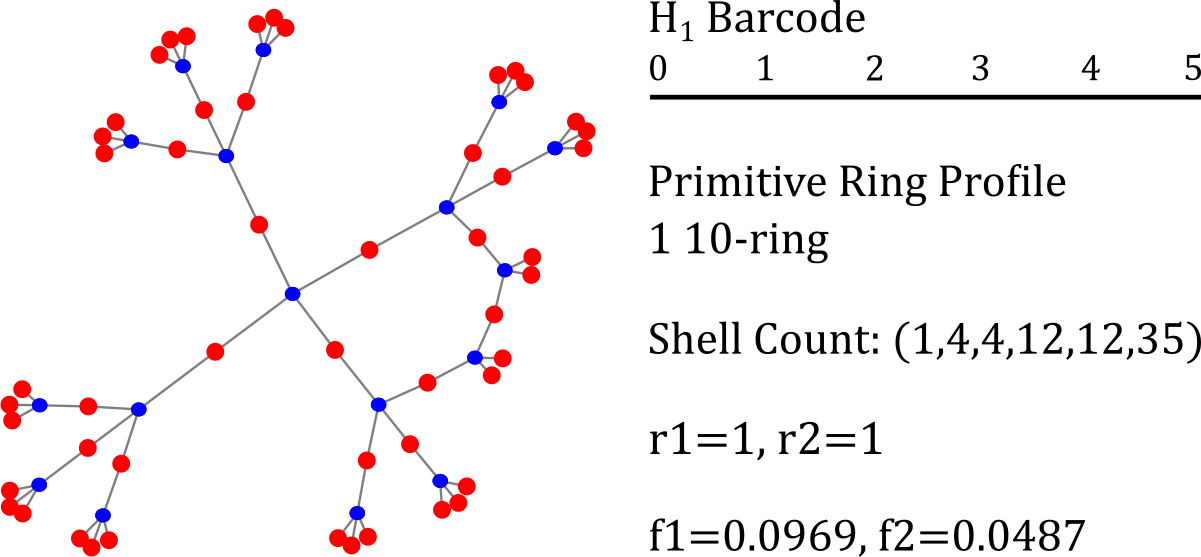}}
\subfigure[]{\label{fig:r5_1_b}
\includegraphics[width=0.3\textwidth]{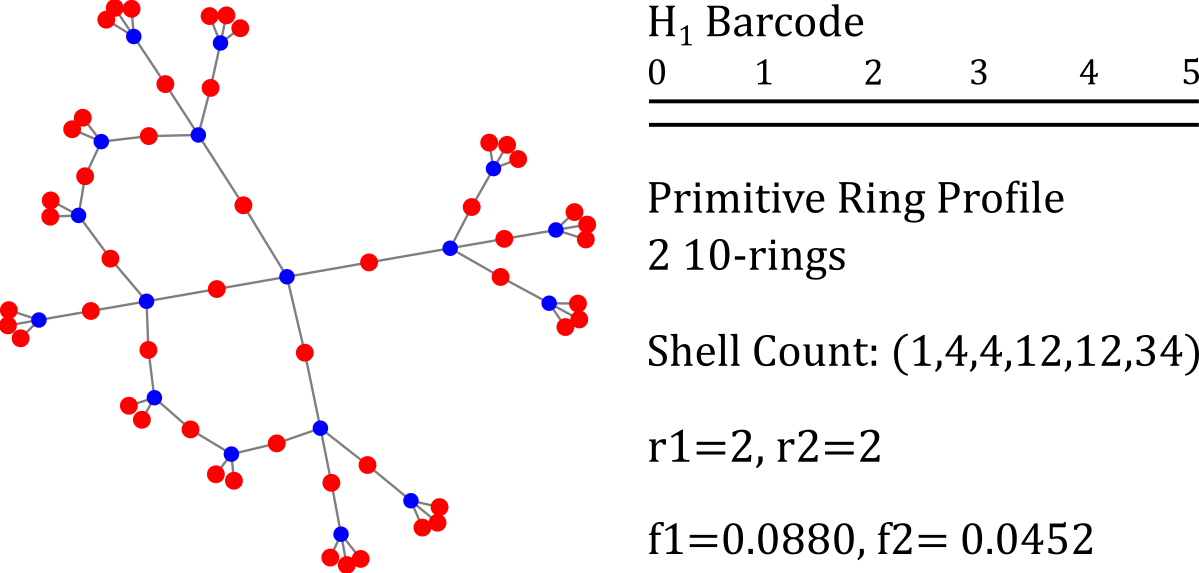}}
\subfigure[]{\label{fig:r5_1_c}
\includegraphics[width=0.3\textwidth]{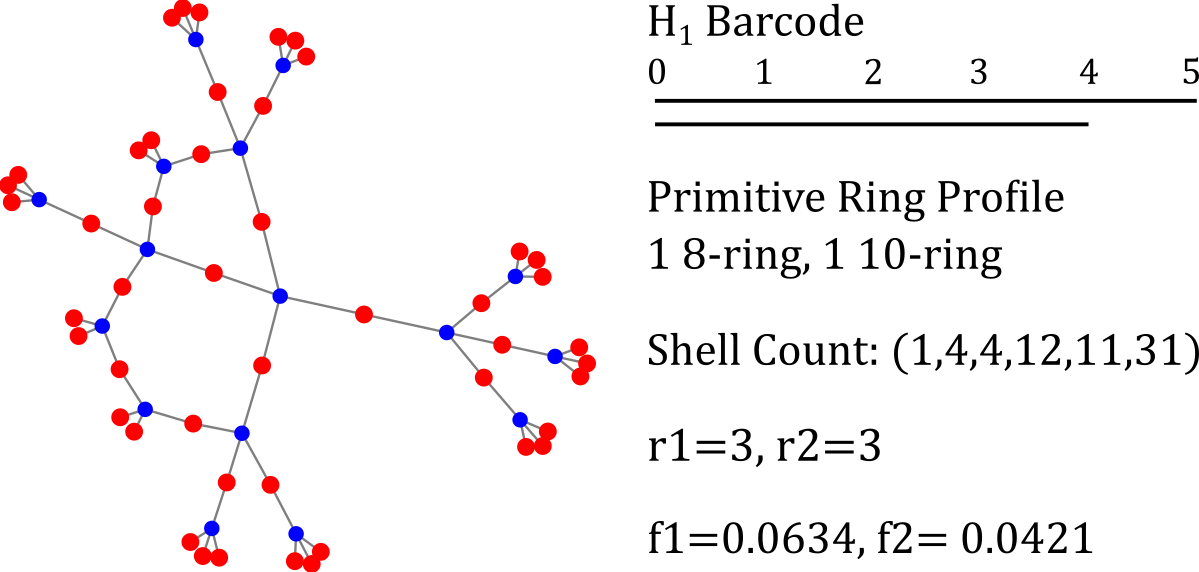}}\\
\subfigure[]{\label{fig:r5_1_d}
\includegraphics[width=0.3\textwidth]{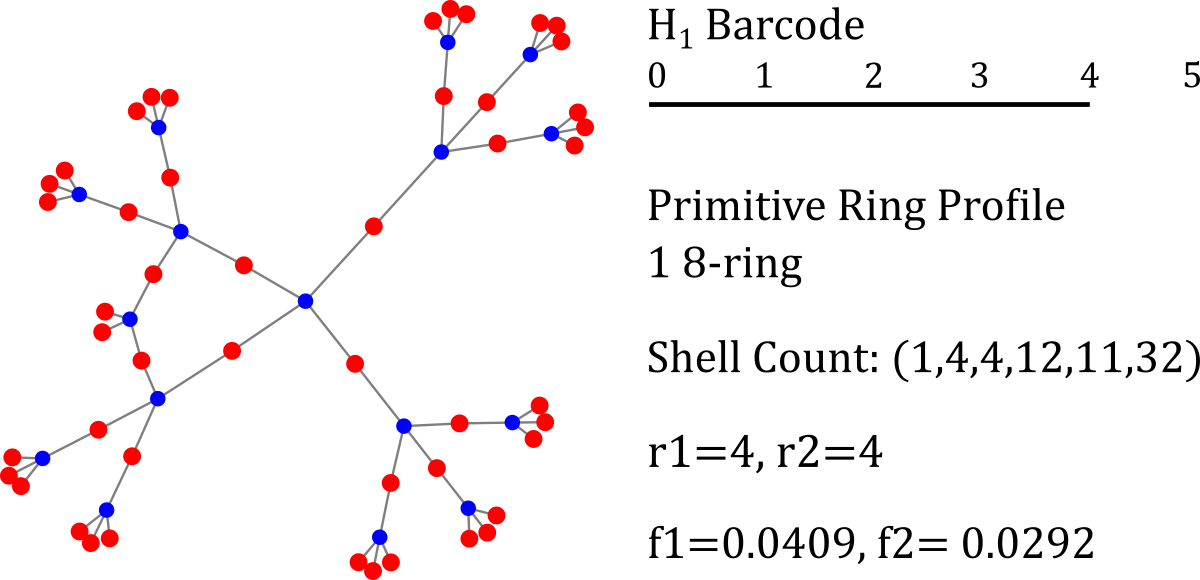}}
\subfigure[]{\label{fig:r5_1_e}
\includegraphics[width=0.3\textwidth]{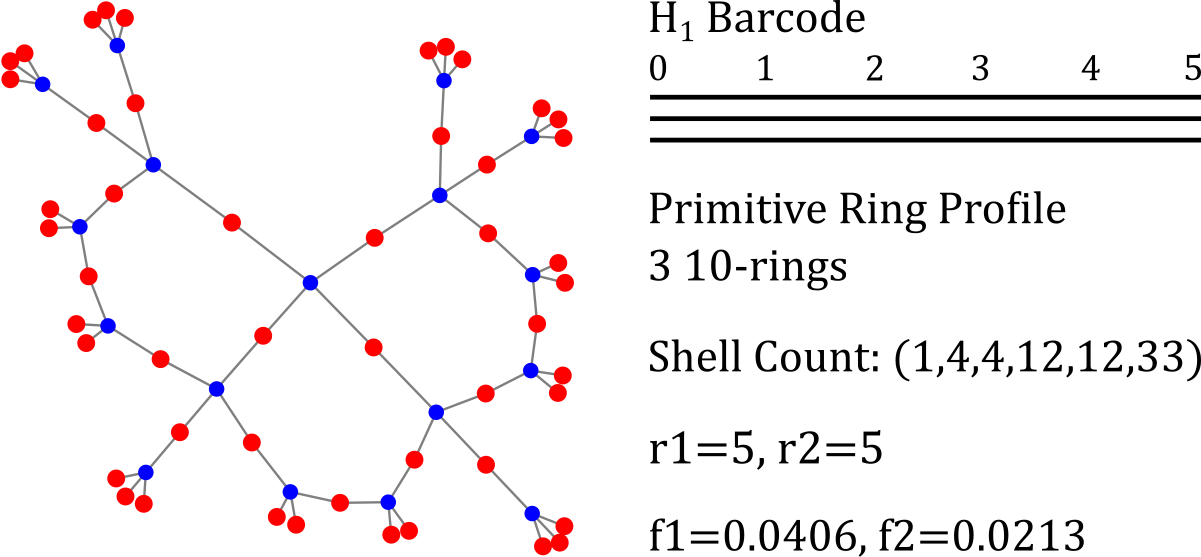}}
\subfigure[]{\label{fig:r5_1_f}
\includegraphics[width=0.3\textwidth]{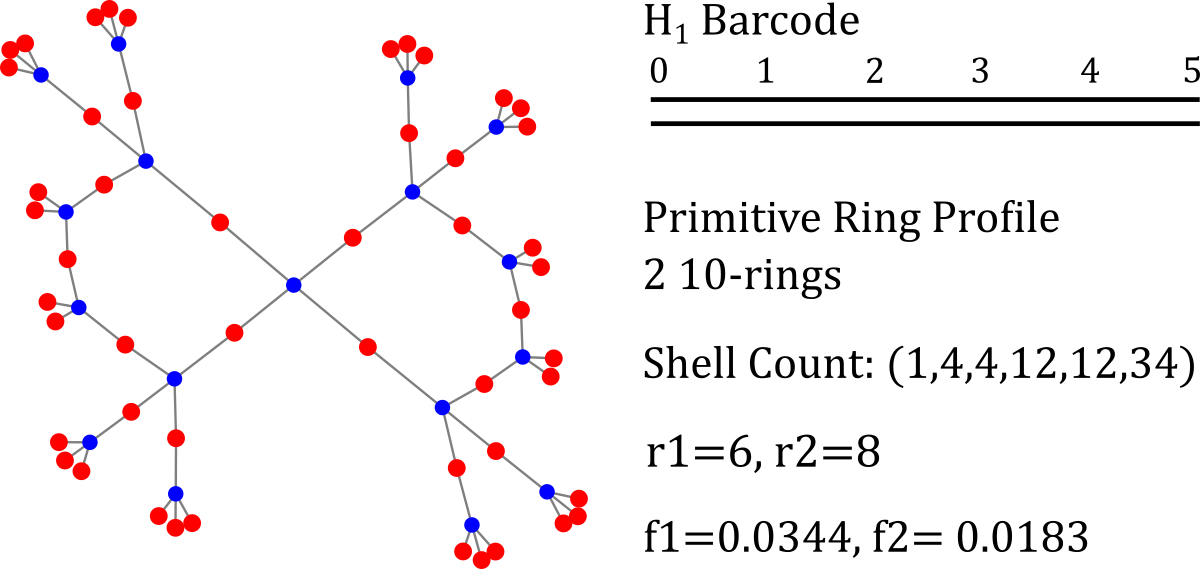}}
\caption{\label{fig:r5_1} The six most frequent graph isomorphism classes at radius $5$ in glasses cooled at a rate of $5 \times 10^{11} \, \mathrm{K / s}$. Graph embeddings where generated using a force-directed layout~\cite{1991fruchterman}.}
\end{figure*}

\begin{figure*}
\centering  
\subfigure[]{\label{fig:r5_2_a}
\includegraphics[width=0.3\textwidth]{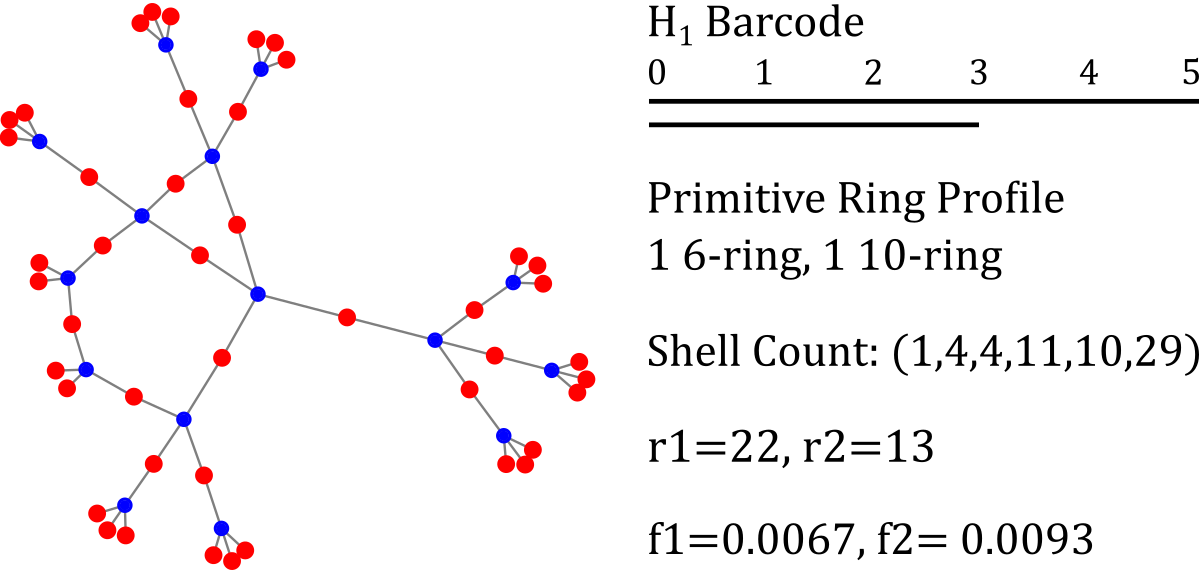}}
\subfigure[]{\label{fig:r5_2_b}
\includegraphics[width=0.3\textwidth]{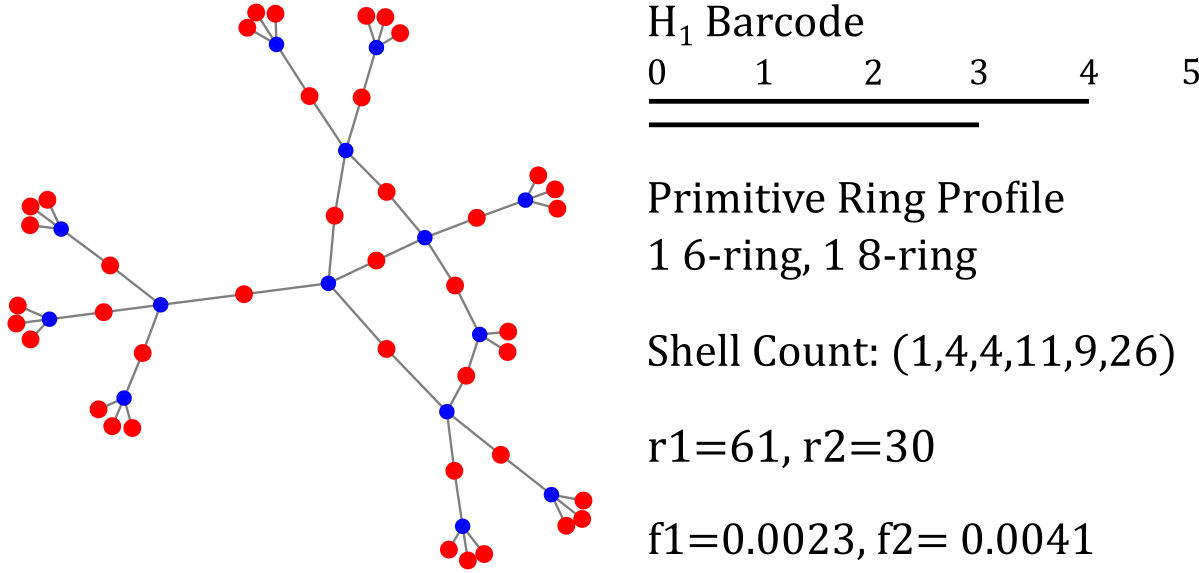}}
\subfigure[]{\label{fig:r5_2_c}
\includegraphics[width=0.3\textwidth]{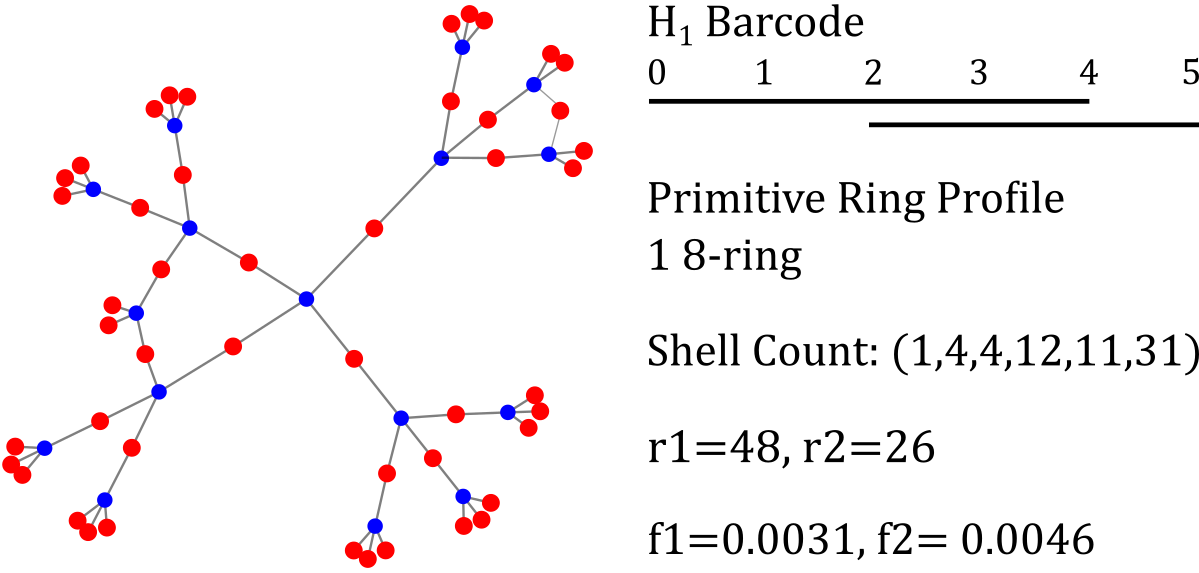}}\\
\subfigure[]{\label{fig:r5_2_d}
\includegraphics[width=0.3\textwidth]{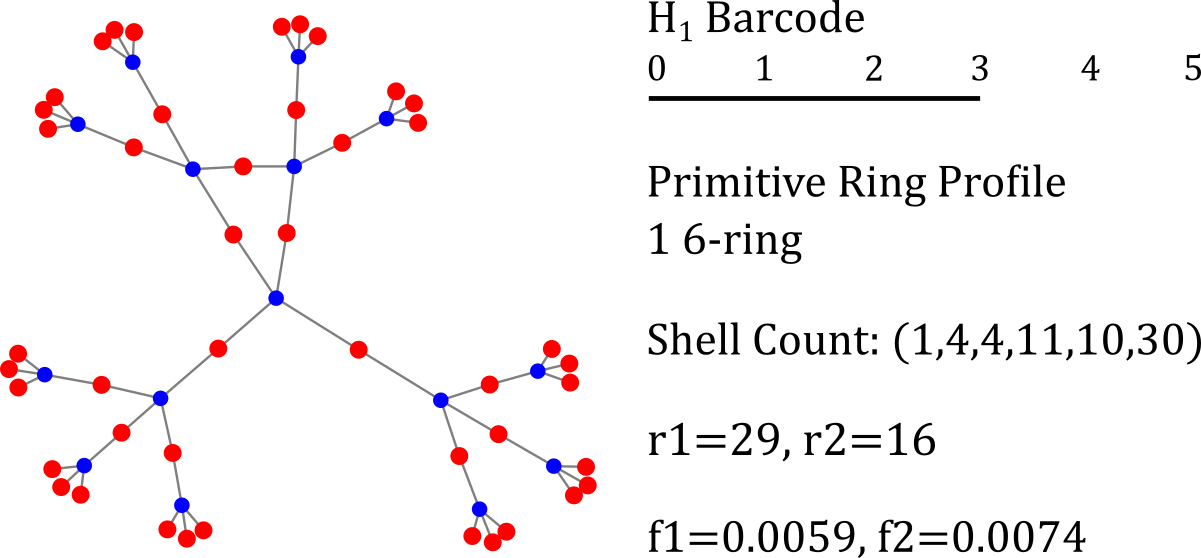}}
\subfigure[]{\label{fig:r5_2_e}
\includegraphics[width=0.3\textwidth]{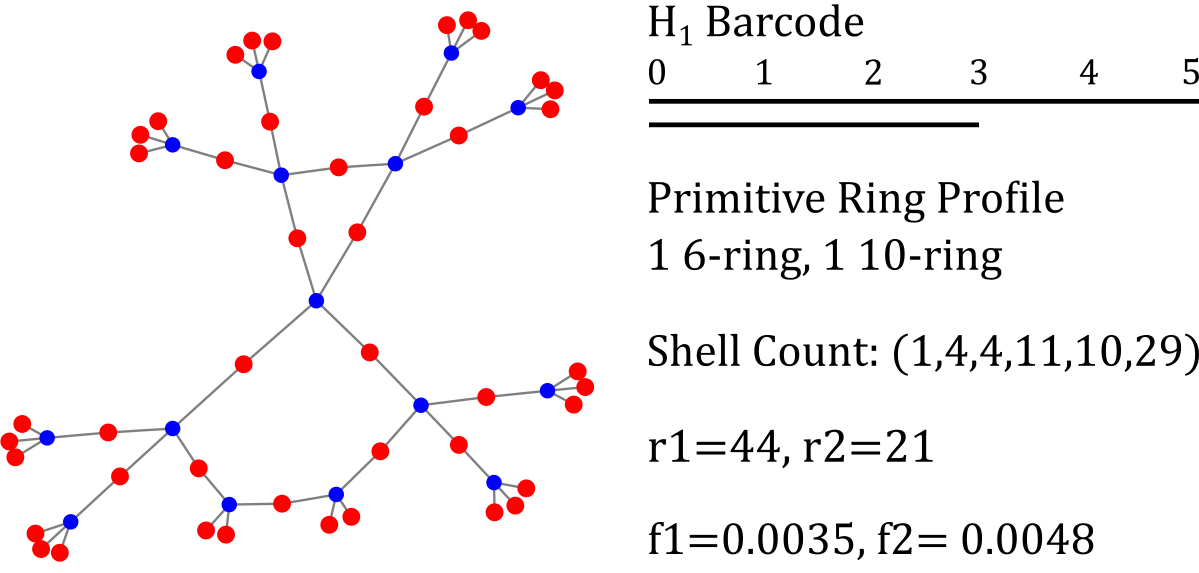}}
\subfigure[]{\label{fig:r5_2_f}
\includegraphics[width=0.3\textwidth]{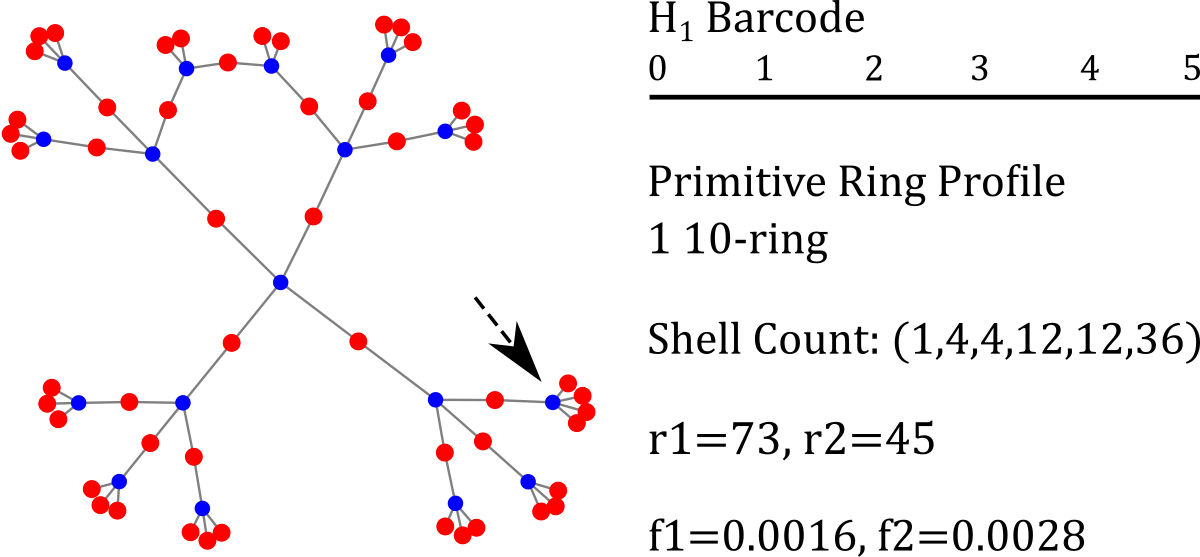}}
\caption{\label{fig:r5_2} The six graph isomorphism classes at radius $5$ that maximize $f_2 - f_1$. }
\end{figure*}
We consider examples of the most common equivalence classes, and the ones that are relatively over-represented in glasses produced with a faster cooling rate. At radius $5,$ Figure~\ref{fig:r5_1} shows the six most common graph isomorphism classes in glasses produced at a cooling rate of $5 \times 10^{11} \, \mathrm{K / s}$, with the corresponding barcode, primitive ring profile, and shell count (there were no coordination errors in these examples). $f_1$ and $r_1$ are the frequency and rank of the isomorphism class at a cooling rate of $5 \times 10^{11} \, \mathrm{K / s}$, and $f_2$ and $r_2$ are the corresponding quantities for a cooling rate of $5 \times 10^{13} \, \mathrm{K / s}$. The ranking is very similar for both conditions, reinforcing that the set of most frequent classes is relatively insensitive to cooling rate. As expected, environments with $10$-rings are well-represented, but (c) and (d) have $8$-rings. More interestingly, several of the most common environments include adjacent rings, i.e., an \ce{Si-O-Si} bond which is shared by $2$ or $3$ $10$-rings as in (b) and (d), or by a $10$- and an $8$-ring as in (c). The adjacency of these rings cannot be detected by ring statistics alone. 

Fig.~\ref{fig:r5_2} shows the six graph isomorphism types that maximize $f_2 - f_1$, i.e., the types most over-represented in glasses produced with the fastest cooling rate. The data reveals that glasses produced with a faster cooling rate have relatively more short rings with $6$ or $8$ atoms, and configurations with adjacent short rings are especially favored. This is consistent with the fact that, at lower quench rates, the number of $12$-rings (with $6$ \ce{Si} atoms) increases while that of other rings decreases~\cite{2015koziatek}. 

Tables~\ref{table:H1_r6}-\ref{table:V_r6} in Appendix~\ref{sec:examples} show similar data for the $H_1$ barcode, primitive ring profile, and coordination profile at radius $r = 6$. For each descriptor, the top ten classes ranked by $f_1$, $f_1 - f_2$, and $f_2 - f_1$ are listed, using the same notation as in the previous paragraph. The descriptors that maximize $f_2 - f_1$ correspond to the peaks in Fig.~\ref{fig:freqs} where the red curve is above the blue one. Examining Table~\ref{table:H1_r6} or~\ref{table:PR_r6} reveals that local environments in glasses produced at $5 \times 10^{11} \, \mathrm{K / s}$ have more rings with $10$ or $12$ atoms ($5$ or $6$ \ce{Si} atoms), while a cooling rate of $5 \times 10^{13} \, \mathrm{K / s}$ yields more $6$- and $8$-rings ($3$ or $4$ \ce{Si} atoms). For example, the early peak in Fig.~\ref{fig:ring_freqs} corresponds to a primitive ring profile with $1$ primitive $8$-ring, $2$ primitive $10$-rings and $2$ primitive $12$-rings. Moreover, glasses produced by faster cooling rates exhibit more coordination errors. In Table~\ref{table:V_r6}, all of the top ten coordination profiles ranked by $f_2 - f_1$ include a silicon in neighbor shell $6$ that is adjacent to $5$ oxygen atoms, and one contains an oxygen adjacent to $3$ silicon atoms.

\section{Conclusion}

The connectivity of the bonds in a network solid is believed to be involved in a variety of physical quantities, e.g., the viscosity and configurational entropy of covalent glass forming materials. The authors suggest that part of the reason that the dependence of the configurational entropy on composition cannot yet be reliably predicted is the lack of a standard approach to quantify the change to the network connectivity with a change in atomic valence. Defining a probability distribution of local atomic environments that is (by hypothesis) characteristic of identically-prepared bond networks establishes a base that can be used to address this issue, and has the further advantage of unifying the treatment of crystalline and amorphous materials.

A closely-related question, and one that is the main subject of this article, is precisely which type of information should be used to classify local atomic environments. Strictly ascribing to the principle that more information is always preferable leads to the situation that all geometric and topological information about the local atomic environments is retained, rendering the approximation of the underlying probability distribution by any finite sampling untenable. Retaining only a subset of the information is therefore a practical necessity, and depending on precisely what information is retained could actually serve to clarify the relationship of the material's structure and properties. For example, the vibrational entropy of a network solid is certainly more closely related to the number of covalent bonds and the constraints they impose than to the precise geometric arrangement of the atoms.

Given this background, four descriptors of local atomic environments were described and applied to silica networks. An appropriately informative descriptor should not only be able to distinguish the different crystalline forms of silica, but should also be able to differentiate silica glasses produced at different cooling rates on the basis of variations in the bond network connectivity. Of the descriptors considered here, the coordination profile and $H_1$ barcode at radius $6$ performed best at these tasks, and were also  faster to compute than the primitive ring profile at the same radius. The efficacy of the coordination profile could be attributed to its sensitivity to coordination defects though; when the subset of local atomic environments without these defects was considered, the $H_1$ barcode outperformed all the other descriptors. Moreover, the $H_1$ barcode and the primitive ring profile could be more readily interpreted than the coordination profile in terms of ring statistics. Overall, the $H_1$ barcode appears to be an efficient, informative, and interpretable descriptor for local atomic environments in silica, and could be invaluable to further advances in our understanding of the structure of amorphous materials.

\begin{acknowledgments}

We would like to thank Dmitriy Morozov for pointing out the connection between the $H_1$ barcode and extended/zigzag persistent homology. B. Schweinhart was partially supported by  NSF Mathematical Sciences Postdoctoral Research Fellowship under award number DMS-1606259, J. K. Mason was supported by the National Science Foundation under Grant No. DMR 1839370, and D. Rodney acknowledges support from LABEX iMUST (ANR-10-LABX- 0064) of Universit\'{e} de Lyon (programme Investissements d’Avenir, ANR-11-IDEX-0007).
\end{acknowledgments}
\bibliography{refs}

\appendix

\section{Computation of the $H_1$ Barcode}
\label{appendix:H1}

The $H_1$ barcode is computed using M\"{o}bius inversion~\cite{1964rota}, which is a generalized version of the inclusion--exclusion principal for partially ordered sets (Ref.~\cite{1975bender} gives an introduction).

Let $S$ be the set of all sub-intervals $\left(a, b\right)$ of $\paren{0, r}$, and let $\left(a, b\right)\leq \left(c, d\right)$ if $\left(a, b\right)\subseteq \left(c, d\right)$. The  M\"{o}bius function of $S$ is defined recursively for all $\left(a, b\right)\leq \left(c, d\right)$ by
\begin{align*}
\mu[\left(a, b\right), \left(a, b\right)] &= 1\\
\mu[\left(a, b\right), \left(c, d\right)] &= -\sum_{\mathclap{\left(a, b\right) \leq \left(e, f\right) \leq \left(c, d\right)}} \mu[\left(a, b\right), \left(e, f\right)].
\end{align*}
It is a mathematical theorem~\cite{1964rota} that if $F$ is a function of the form
\[F[\left(c, d\right)] = \sum_{\mathclap{\left(a, b\right) \leq \left(c, d\right)}} G[\left(a, b\right)]\]
then $G$ can be recovered by the formula
\[G[\left(c, d\right)] =\sum_{\mathclap{\left(a, b\right) \leq \left(c, d\right)}} F[\left(a, b\right)] \mu[\left(a, b\right), \left(c, d\right)].\]
In particular, if $F[\left(a, b\right)]$ is given as in Eq.~\ref{eq_F}, then $G[\left(a, b\right)]$ is the number of intervals of the form $\left(a, b\right)$ in the $H_1$ barcode. Practically, $F[\left(a, b\right)]$ is computed using the Euler characteristic as in Eq.~\ref{eqn_chi} before applying the previous formula to compute the number of intervals.

\section{The $H_1$ Barcode of Perfectly Coordinated Environments} 
\label{sec:perfectH1}
We show that the information in the shell count is equivalent to that of the endpoints of the $H_1$ barcode intervals for perfectly coordinated environments, where a perfectly coordinated environment is a bipartite rooted graph so vertices in even shells have degree $d_0$ and vertices in odd shells have degree $d_1.$ For the perfectly coordinated silica environments considered in the text, $d_0=4$ and $d_1=2.$ 

\begin{proposition}
Let $G$ be a perfectly coordinated local atomic environment of radius $r.$ The endpoints of the $H_1$ barcode intervals of $G$ can be computed in terms of the shell count, and visa versa.
\end{proposition}
\begin{proof} 
First, note that the number of  $H_1$ barcode intervals whose endpoints are $\leq r_0$ equals (by definition) $F\paren{0,r_0},$ the rank of the first homology group of the local atomic environment of radius $r_0.$ As such, knowledge of the endpoints is equivalent to knowledge of $F\paren{0,r_0}$ for all $r_0\leq r.$ Assume that we know  $F\paren{0,r_0}$ for all $r_0\leq r$ and suppose by induction that we have computed the shell count at radius $r_0$ ($\#\text{atoms}\paren{r_0}$) as well as the number of bonds between the atoms in shells $r_0$ and $r_0-1$ ($\#\text{bonds}\paren{r_0,r_0-1}$) for all $r_0<r.$ The environment is bipartite, so atoms in shell $r-1$ can share bonds with atoms in shells $r$ or $r-2$ but not atoms in shell $r-1.$ It follows that 
\[\#\text{bonds}\paren{r,r-1}=d\,\#\text{atoms}\paren{r-1}-\#\text{bonds}\paren{r-1,r-2}\,\]
where $d=d_0$ if $r-1$ is even and $d=d_1$ otherwise.
\[ \#\text{bonds}\paren{G}=\sum_{r_0\leq r}\#\text{bonds}\paren{r_0,r_0-1}\]
and we can use Eqn.~\ref{eqn_chi} to find
\[\#\text{atoms}\paren{G}=1+ \#\text{bonds}\paren{G}-F\paren{0,r_0}\,.\]
The shell count at radius $r$ is then
\[\#\text{atoms}\paren{r}=\#\text{atoms}\paren{G}-\sum_{r_0<r} \#\text{atoms}\paren{r_0}\,.\]
A similar argument shows that we can compute $F\paren{0,r_0}$ for all $r_0\leq r$ given knowledge of the shell counts.
\end{proof}

\section{Molecular Dynamics Methodology}
\label{sec:MD}

Molecular dynamics simulations of silica glasses were performed using the classical two-body potential developed by van Beest, Kramer and van Santen (BKS)~\cite{1955beest}, modified by Carre\'{e} et al.~\cite{2007carre} to replace the long-range Coulombic interaction by a Wolf truncation. The parameters of the potential are given in Ref.~\cite{2012mantisi}. The glasses considered here contained $3,000$ atoms ($1,000$ \ce{Si} and $2,000$ \ce{O} atoms) in cubic simulation cells with periodic boundary conditions and a density of $2.2 \, \mathrm{g / cm^3}$. They were produced from melts equilibrated at $5200 \, \mathrm{K}$ and quenched at constant volume to $10 \, \mathrm{K}$ at rates between $5 \times 10^{11} \, \mathrm{K / s}$ and $5 \times 10^{13} \, \mathrm{K / s}$. The simulations used a timestep of $1 \, \mathrm{fs}$ and an Andersen thermostat~\cite{1980andersen} which re-assigned atomic velocities with a probability of $0.001$ per timestep. The quenches were followed by energy minimizations to obtain inherent structures before further analyses.

\clearpage
\onecolumngrid

\section{Ranked Examples}
\label{sec:examples}

\begin{table}[H]
\centering
\begin{tabular}{l|c|c|c}
Sorted by & $f_1$ & $f_1 - f_2$ & $f_2 - f_1$\\
\hline
1 &$ 2\times(0,5),(0,6),4\times(2,6)$&$ 2\times(0,5),(0,6),4\times(2,6)$&$ (0,4),(0,5),(0,6),(2,5),4\times(2,6)$\\
& $f_1=0.033, f_2=0.017$& $f_1=0.033, f_2=0.017$& $f_1=0.004, f_2=0.006$\\
& $r_1=1, r_2=2$& $r_1=1, r_2=2$& $r_1=58, r_2=31$\\
\hline
2 &$ 2\times(0,5),(0,6),3\times(2,6)$&$ 2\times(0,5),(0,6),3\times(2,6)$&$ 2\times(0,4),(0,5),(2,5),2\times(2,6)$\\
& $f_1=0.031, f_2=0.015$& $f_1=0.031, f_2=0.015$& $f_1=0.003, f_2=0.005$\\
& $r_1=2, r_2=3$& $r_1=2, r_2=3$& $r_1=66, r_2=43$\\
\hline
3 &$ (0,4),(0,5),(0,6),3\times(2,6)$&$ (0,5),2\times(0,6),4\times(2,6)$&$ 2\times(0,4),(0,5),(2,5),3\times(2,6)$\\
& $f_1=0.027, f_2=0.018$& $f_1=0.021, f_2=0.010$& $f_1=0.002, f_2=0.003$\\
& $r_1=3, r_2=1$& $r_1=5, r_2=15$& $r_1=103, r_2=61$\\
\hline
4 &$ 2\times(0,5),(0,6),5\times(2,6)$&$ (0,5),2\times(0,6),3\times(2,6)$&$ (0,4),(0,5),(0,6),4\times(2,6),(4,6)$\\
& $f_1=0.022, f_2=0.012$& $f_1=0.019, f_2=0.009$& $f_1=0.003, f_2=0.004$\\
& $r_1=4, r_2=7$& $r_1=8, r_2=17$& $r_1=78, r_2=50$\\
\hline
5 &$ (0,5),2\times(0,6),4\times(2,6)$&$ 2\times(0,5),(0,6),5\times(2,6)$&$ (0,4),2\times(0,5),2\times(2,5),3\times(2,6)$\\
& $f_1=0.021, f_2=0.010$& $f_1=0.022, f_2=0.012$& $f_1=0.002, f_2=0.003$\\
& $r_1=5, r_2=15$& $r_1=4, r_2=7$& $r_1=118, r_2=76$\\
\hline
6 &$ (0,4),(0,5),(0,6),2\times(2,6)$&$ 2\times(0,5),(0,6),2\times(2,6)$&$ (0,4),2\times(0,5),(2,5),3\times(2,6),(4,6)$\\
& $f_1=0.020, f_2=0.012$& $f_1=0.017, f_2=0.008$& $f_1=0.001, f_2=0.002$\\
& $r_1=6, r_2=6$& $r_1=11, r_2=22$& $r_1=151, r_2=93$\\
\hline
7 &$ (0,4),(0,5),(0,6),4\times(2,6)$&$ (0,4),(0,5),(0,6),3\times(2,6)$&$ (0,4),2\times(0,5),(2,5),4\times(2,6),(4,6)$\\
& $f_1=0.020, f_2=0.014$& $f_1=0.027, f_2=0.018$& $f_1=0.001, f_2=0.002$\\
& $r_1=7, r_2=4$& $r_1=3, r_2=1$& $r_1=224, r_2=117$\\
\hline
8 &$ (0,5),2\times(0,6),3\times(2,6)$&$ 3\times(0,5),4\times(2,6)$&$ (0,4),2\times(0,5),(2,5),2\times(2,6),(4,6)$\\
& $f_1=0.019, f_2=0.009$& $f_1=0.019, f_2=0.010$& $f_1=0.001, f_2=0.002$\\
& $r_1=8, r_2=17$& $r_1=9, r_2=12$& $r_1=145, r_2=94$\\
\hline
9 &$ 3\times(0,5),4\times(2,6)$&$ 3\times(0,5),3\times(2,6)$&$ (0,3),(0,5),(0,6),3\times(2,6)$\\
& $f_1=0.019, f_2=0.010$& $f_1=0.017, f_2=0.008$& $f_1=0.002, f_2=0.004$\\
& $r_1=9, r_2=12$& $r_1=12, r_2=21$& $r_1=83, r_2=56$\\
\hline
10 &$ (0,4),2\times(0,5),3\times(2,6)$&$ (0,4),(0,5),(0,6),2\times(2,6)$&$ 2\times(0,4),(0,5),2\times(2,5),2\times(2,6)$\\
& $f_1=0.018, f_2=0.013$& $f_1=0.020, f_2=0.012$& $f_1=0.001, f_2=0.002$\\
& $r_1=10, r_2=5$& $r_1=6, r_2=6$& $r_1=237, r_2=143$\\\end{tabular}
\caption{\label{table:H1_r6} The $H_1$ barcodes at radius $6$ that maximize $f_1$, $f_1 - f_2$, and $f_2 - f_1$. $f_1$ and $r_1$ are the frequency and rank of the isomorphism class at a cooling rate of $5 \times 10^{11} \, \mathrm{K / s}$, and $f_2$ and $r_2$ are the corresponding quantities for a cooling rate of $5 \times 10^{13} \, \mathrm{K / s}$. See Section~\ref{sec:ranked} for details and analysis.}
\end{table}

\begin{table}
\centering
\begin{tabular}{l|c|c|c}
Sorted by & $f_1$ & $f_1-f_2$ & $f_2-f_1$\\
\hline
1 &$ 2\; 10\text{-rings},\,3\; 12\text{-rings}$&$ 2\; 10\text{-rings},\,3\; 12\text{-rings}$&$ 1\; 6\text{-ring},\,1\; 10\text{-ring},\,2\; 12\text{-rings}$\\
& $f_1=0.049, f_2=0.035$& $f_1=0.049, f_2=0.035$& $f_1=0.004, f_2=0.008$\\
& $r_1=1, r_2=1$& $r_1=1, r_2=1$& $r_1=61, r_2=33$\\
\hline
2 &$ 2\; 10\text{-rings},\,4\; 12\text{-rings}$&$ 2\; 10\text{-rings},\,4\; 12\text{-rings}$&$ 2\; 8\text{-rings},\,1\; 10\text{-ring},\,1\; 12\text{-ring}$\\
& $f_1=0.042, f_2=0.029$& $f_1=0.042, f_2=0.029$& $f_1=0.005, f_2=0.008$\\
& $r_1=2, r_2=3$& $r_1=2, r_2=3$& $r_1=46, r_2=38$\\
\hline
3 &$ 3\; 10\text{-rings},\,2\; 12\text{-rings}$&$ 3\; 10\text{-rings},\,3\; 12\text{-rings}$&$ 1\; 6\text{-ring},\,1\; 10\text{-ring},\,1\; 12\text{-ring}$\\
& $f_1=0.039, f_2=0.030$& $f_1=0.037, f_2=0.027$& $f_1=0.003, f_2=0.006$\\
& $r_1=3, r_2=2$& $r_1=4, r_2=5$& $r_1=67, r_2=48$\\
\hline
4 &$ 3\; 10\text{-rings},\,3\; 12\text{-rings}$&$ 2\; 10\text{-rings},\,5\; 12\text{-rings}$&$ 1\; 6\text{-ring},\,1\; 10\text{-ring},\,3\; 12\text{-rings}$\\
& $f_1=0.037, f_2=0.027$& $f_1=0.028, f_2=0.018$& $f_1=0.003, f_2=0.005$\\
& $r_1=4, r_2=5$& $r_1=8, r_2=14$& $r_1=68, r_2=53$\\
\hline
5 &$ 2\; 10\text{-rings},\,2\; 12\text{-rings}$&$ 1\; 10\text{-ring},\,5\; 12\text{-rings}$&$ 2\; 8\text{-rings},\,1\; 10\text{-ring},\,2\; 12\text{-rings}$\\
& $f_1=0.033, f_2=0.027$& $f_1=0.026, f_2=0.017$& $f_1=0.006, f_2=0.008$\\
& $r_1=5, r_2=6$& $r_1=10, r_2=16$& $r_1=44, r_2=35$\\
\hline
6 &$ 1\; 8\text{-ring},\,1\; 10\text{-ring},\,3\; 12\text{-rings}$&$ 3\; 10\text{-rings},\,2\; 12\text{-rings}$&$ 1\; 6\text{-ring},\,2\; 10\text{-rings},\,2\; 12\text{-rings}$\\
& $f_1=0.028, f_2=0.027$& $f_1=0.039, f_2=0.030$& $f_1=0.002, f_2=0.005$\\
& $r_1=6, r_2=7$& $r_1=3, r_2=2$& $r_1=75, r_2=59$\\
\hline
7 &$ 1\; 8\text{-ring},\,2\; 10\text{-rings},\,2\; 12\text{-rings}$&$ 1\; 10\text{-ring},\,6\; 12\text{-rings}$&$ 1\; 6\text{-ring},\,2\; 10\text{-rings},\,1\; 12\text{-ring}$\\
& $f_1=0.028, f_2=0.028$& $f_1=0.018, f_2=0.011$& $f_1=0.003, f_2=0.005$\\
& $r_1=7, r_2=4$& $r_1=19, r_2=26$& $r_1=69, r_2=55$\\
\hline
8 &$ 2\; 10\text{-rings},\,5\; 12\text{-rings}$&$ 1\; 10\text{-ring},\,4\; 12\text{-rings}$&$ 2\; 8\text{-rings},\,2\; 10\text{-rings},\,2\; 12\text{-rings}$\\
& $f_1=0.028, f_2=0.018$& $f_1=0.027, f_2=0.020$& $f_1=0.003, f_2=0.005$\\
& $r_1=8, r_2=14$& $r_1=9, r_2=12$& $r_1=66, r_2=54$\\
\hline
9 &$ 1\; 10\text{-ring},\,4\; 12\text{-rings}$&$ 2\; 10\text{-rings},\,2\; 12\text{-rings}$&$ 1\; 6\text{-ring},\,1\; 8\text{-ring},\,1\; 10\text{-ring},\,1\; 12\text{-ring}$\\
& $f_1=0.027, f_2=0.020$& $f_1=0.033, f_2=0.027$& $f_1=0.001, f_2=0.003$\\
& $r_1=9, r_2=12$& $r_1=5, r_2=6$& $r_1=97, r_2=73$\\
\hline
10 &$ 1\; 10\text{-ring},\,5\; 12\text{-rings}$&$ 3\; 10\text{-rings},\,4\; 12\text{-rings}$&$ 1\; 6\text{-ring},\,1\; 8\text{-ring},\,1\; 10\text{-ring},\,2\; 12\text{-rings}$\\
& $f_1=0.026, f_2=0.017$& $f_1=0.022, f_2=0.017$& $f_1=0.001, f_2=0.003$\\
& $r_1=10, r_2=16$& $r_1=15, r_2=15$& $r_1=117, r_2=93$\\
\end{tabular}
\caption{\label{table:PR_r6} The primitive ring profiles at radius $6$ that maximize $f_1$, $f_1 - f_2$, and $f_2 - f_1$.}
\end{table}

\clearpage
\begin{table}
\centering
\begin{tabular}{l|c|c|c}
Sorted by & $f_1$ & $f_1-f_2$ & $f_2-f_1$\\
\hline
1 &$ (1,4,4,12,12,34,27)$&$ (1,4,4,12,12,34,27)$&$ (1,4,4,12,11,30,24^{\ast})$\\
& $f_1=0.038, f_2=0.014$& $f_1=0.038, f_2=0.014$& $f_1=0.002, f_2=0.003$\\
& $r_1=1, r_2=1$& $r_1=1, r_2=1$& $r_1=91, r_2=38$\\
\hline
2 &$ (1,4,4,12,12,34,28)$&$ (1,4,4,12,12,34,28)$&$ (1,4,4,12,11,30,25^{\ast})$\\
& $f_1=0.034, f_2=0.012$& $f_1=0.034, f_2=0.012$& $f_1=0.002, f_2=0.004$\\
& $r_1=2, r_2=5$& $r_1=2, r_2=5$& $r_1=84, r_2=34$\\
\hline
3 &$ (1,4,4,12,12,33,26)$&$ (1,4,4,12,12,33,26)$&$ (1,4,4,12,11,31,26^{\ast})$\\
& $f_1=0.032, f_2=0.012$& $f_1=0.032, f_2=0.012$& $f_1=0.002, f_2=0.004$\\
& $r_1=3, r_2=3$& $r_1=3, r_2=3$& $r_1=82, r_2=37$\\
\hline
4 &$ (1,4,4,12,11,31,26)$&$ (1,4,4,12,12,33,27)$&$ (1,4,4,12,11,30,26^{\ast})$\\
& $f_1=0.028, f_2=0.012$& $f_1=0.027, f_2=0.010$& $f_1=0.001, f_2=0.003$\\
& $r_1=4, r_2=4$& $r_1=6, r_2=10$& $r_1=101, r_2=47$\\
\hline
5 &$ (1,4,4,12,12,34,26)$&$ (1,4,4,12,12,34,26)$&$ (1,4,4,12,12,33,25^{\ast})$\\
& $f_1=0.027, f_2=0.010$& $f_1=0.027, f_2=0.010$& $f_1=0.001, f_2=0.003$\\
& $r_1=5, r_2=7$& $r_1=5, r_2=7$& $r_1=96, r_2=45$\\
\hline
6 &$ (1,4,4,12,12,33,27)$&$ (1,4,4,12,11,31,26)$&$ (1,4,4,12,12,34,27^{\ast})$\\
& $f_1=0.027, f_2=0.010$& $f_1=0.028, f_2=0.012$& $f_1=0.003, f_2=0.004$\\
& $r_1=6, r_2=10$& $r_1=4, r_2=4$& $r_1=56, r_2=27$\\
\hline
7 &$ (1,4,4,12,11,30,25)$&$ (1,4,4,12,12,35,28)$&$ (1,4,4,12,12,33,26^{\ast})$\\
& $f_1=0.026, f_2=0.013$& $f_1=0.022, f_2=0.007$& $f_1=0.003, f_2=0.004$\\
& $r_1=7, r_2=2$& $r_1=9, r_2=13$& $r_1=66, r_2=32$\\
\hline
8 &$ (1,4,4,12,12,33,25)$&$ (1,4,4,12,12,35,29)$&$ (1,4,4,12,11,29,24^{\ast})$\\
& $f_1=0.024, f_2=0.010$& $f_1=0.020, f_2=0.006$& $f_1=0.001, f_2=0.002$\\
& $r_1=8, r_2=8$& $r_1=13, r_2=19$& $r_1=167, r_2=73$\\
\hline
9 &$ (1,4,4,12,12,35,28)$&$ (1,4,4,12,12,33,25)$&$ (1,4,4,12,12,34,28^{\ast})$\\
& $f_1=0.022, f_2=0.007$& $f_1=0.024, f_2=0.010$& $f_1=0.003, f_2=0.004$\\
& $r_1=9, r_2=13$& $r_1=8, r_2=8$& $r_1=64, r_2=31$\\
\hline
10 &$ (1,4,4,12,11,31,27)$&$ (1,4,4,12,12,34,29)$&$ (1,4,4,12,11,31,25^{\ast})$\\
& $f_1=0.021, f_2=0.009$& $f_1=0.019, f_2=0.006$& $f_1=0.002, f_2=0.003$\\
& $r_1=10, r_2=11$& $r_1=14, r_2=18$& $r_1=92, r_2=49$\\\end{tabular}
\caption{\label{table:V_r6} The coordination profiles at radius $6$ that maximize $f_1$, $f_1 - f_2$, and $f_2 - f_1$. The shell count is shown, and $\ast$ indicates the presence of one atom of valence $5$.}
\end{table}

\end{document}